\documentclass[pdflatex,sn-mathphys-num]{sn-jnl}

\usepackage{graphicx}%
\usepackage{multirow}%
\usepackage{amsmath,amssymb,amsfonts}%
\usepackage{amsthm}%
\usepackage{mathrsfs}%
\usepackage[title]{appendix}%
\usepackage{xcolor}%
\usepackage{textcomp}%
\usepackage{manyfoot}%
\usepackage{booktabs}%
\usepackage{algorithm}%
\usepackage{algorithmicx}%
\usepackage{algpseudocode}%
\usepackage{listings}%
\usepackage[qm]{qcircuit}
\usepackage{braket}         
\usepackage{subfig}
\usepackage{todonotes}

\newtheorem{theorem}{Theorem}
\newtheorem{definition}{Definition}
\newcommand{\cprod}[2]{\overset{\curvearrowleft}{\prod^{#2}_{#1}}}
\newcommand{\oprod}[2]{\overset{\curvearrowright}{\prod^{#2}_{#1}}}
\newcommand{\mat}[1]{\left(\begin{matrix}#1\end{matrix}\right)}

\newcounter{mysfig}
\counterwithin{mysfig}{figure}
\renewcommand\themysfig{\thefigure(\alph{mysfig})}
\makeatletter
\newcommand\Scaption[1]{%
	\refstepcounter{mysfig}%
	\vskip.5\abovecaptionskip
	\sbox\@tempboxa{\small\themysfig~#1}%
	\ifdim \wd\@tempboxa >\hsize
	\small\themysfig~#1\par
	\else
	\global \@minipagefalse
	\hb@xt@\hsize{\hfil\box\@tempboxa\hfil}%
	\fi
	\vskip\belowcaptionskip}
\makeatother

\definecolor{darkgreen}{rgb}{0,0.7,0}
\usetikzlibrary{math}
\usetikzlibrary {decorations.pathmorphing}

\begin{document}
	
	\title{A quantum walk inspired model for distributed computing on arbitrary graphs}
	
	\author*[]{\fnm{Mathieu} \sur{Roget}}\email{mathieu.roget@lis-lab.fr}
	
	\author[]{\fnm{Giuseppe} \sur{Di Molfetta}}\email{giuseppe.dimolfetta@lis-lab.fr}
	
	\affil[]{\orgname{Aix-Marseille Université, Université de Toulon, CNRS, LIS}, \city{Marseille}, \state{France}}

	\abstract{A discrete time quantum walk is known to be the single-particle sector of a quantum cellular automaton. For a long time, these models have interested the community for their nice properties such as locality or translation invariance. This work introduces a model of distributed computation for arbitrary graphs inspired by quantum cellular automata. As a by-product, we show how this model can reproduce the dynamic of a quantum walk on graphs. In this context, we investigate the communication cost for two interaction schemes. Finally, we explain how this particular quantum walk can be applied to solve the search problem and present numerical results on different types of topologies.}

	
	\maketitle
	
	\section*{Introduction}\label{sec:intro}
	Quantum Walks~(QW), from a mathematical point of view, coincide with the single-particle sector of quantum cellular automata (QCA), namely a spatially distributed network of local quantum gates. Usually defined on a $d$-dimensional grid of cell, they are known to be capable of universal computation~\cite{arrighi2019overview}. Both quantum walks and quantum cellular automata, with their beginnings in digital simulations of fundamental physics~\cite{di2013quantum, bisio2016quantum}, come into their own in algorithmic search and optimization applications~\cite{santha2008quantum}. Searching using QW has been extensively studied in the past decades, with a wide range of applications, including optimization~\cite{slate2021quantum} and machine learning~\cite{melnikov2019predicting}. On the other hand, search algorithms as a field independent of quantum walks have recently been used, as subroutines, to solve distributed computational tasks~\cite{gall2018quantum, le2018sublinear, izumi2019quantum}. However, in all these examples, some global information is supposed to be known by each node of the network, such as the size, and usually the network is not anonymous, namely every node has a unique label. Note that, the absence of anonymity and the quantum properties of the algorithm successfully solves the incalculability of certain problems such as the leader election problem. However, encoding global information within a quantum state is generally problematic. In order to address such issue, here we first introduce a new QW-based scheme on arbitrary graphs for rephrasing search algorithms. Then we move to the multi-qubits generalization, which successfully implements a QCA-based distributed anonymous protocol for searching problems, avoiding any use of global information.
	
	\paragraph{Contribution}
	Section \ref{sec:model} of this work introduces a Quantum Walk model well suited to indifferently search for a node or edge in arbitrary graphs. In this model the walker's amplitudes are defined onto the graph's edges and ensure 2-dimensional coin everywhere. We detail how this Quantum Walk can be used to search a node or an edge and we show examples of this Quantum Walk on several graphs. In this first part, the Quantum Walk is introduced formally as a discrete dynamical system. In Section \ref{sec:distrib} we move to the multi-particle sector, allowing many quantum states-dynamics over the network, based on the previous model, and leading to two distributed searching protocols. We consider here two families of interaction graphs: all-to-all and cyclic. The implementations proposed conserves the graph locality and does not require a node or an edge to have global information such as the graph size. The nodes (and edges) do not have a unique label, and no leader is needed. Finally, Section \ref{sec:application} shows how to apply our model to the search of edges or nodes. Numerical experiments for grids, hypercubes, complete graphs, and random scale-free graphs are provided.
	
	\section{Model of Quantum Walk on graphs}\label{sec:model}

This section introduces our model of quantum walk on graphs and compare it to other similar models.

\subsection{The model}
We consider an undirected connected graph $G = (V,E)$, where $V$ is the set of vertices and $E$ the set of edges. We define the walker's position on the graph's edges and a coin register of dimension two (either $+$ or $-$). The full state of the walker at step $t$ is noted~:
$$
\ket{\Phi_t} = \sum_{(u,v)\in E} \psi_{u,v}^+\ket{(u,v)}\ket{+} + \psi_{u,v}^-\ket{(u,v)}\ket{-}.
$$

The graph is undirected so we indifferently use $\ket{(u,v)}$ and $\ket{(v,u)}$ to name the same state. Similarly for the complex amplitudes $ \psi_{u,v}^+(t) =  \psi_{v,u}^-(t)$. We also introduce a polarity for every edge of $G$. For each of them, we have a polarity function $\sigma$, such that~:
$
\forall (u,v) \in E, \; \sigma(u,v) \in \{+,-\} \text{ and } \sigma(u,v) \neq \sigma(v,u).
$
Figure \ref{fig:tikzedge} illustrates how amplitudes and polarity are placed with edge $(u,v)$ of polarity $\sigma(u,v) = +$. 
\begin{figure}[h]
	\centering
	\subfloat[$\sigma(u,v) = +$]{\includegraphics[width=0.4\linewidth]{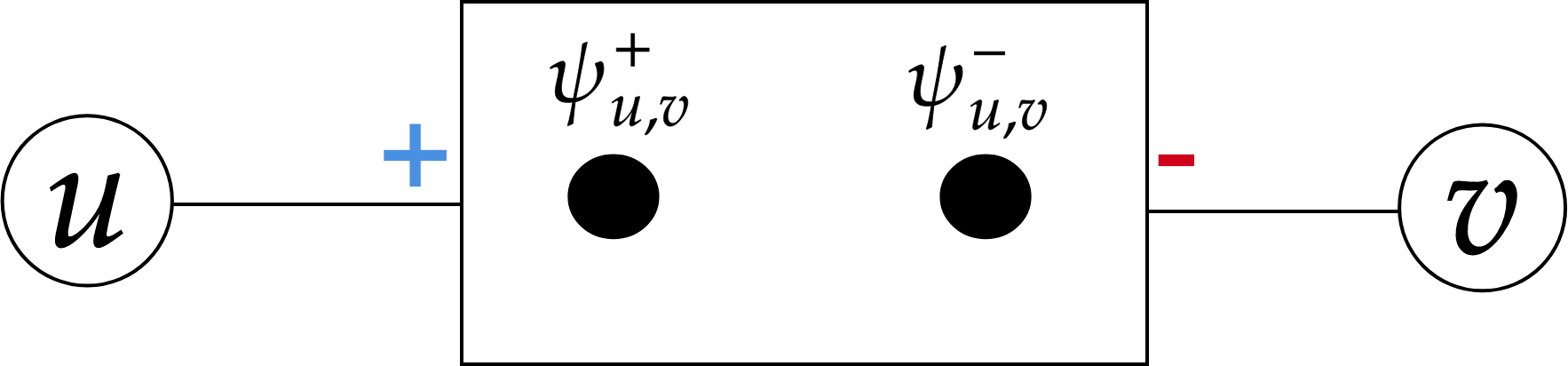}}\qquad
	\subfloat[$\sigma(u,v) = -$]{\includegraphics[width=0.4\linewidth]{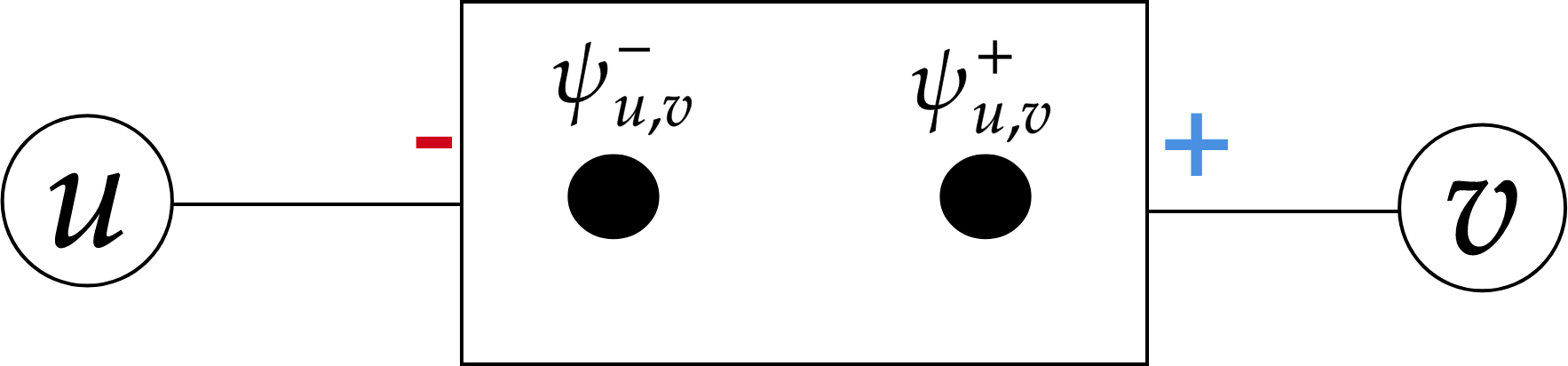}}\\
	\caption{The edge $(u,v)$ and how are placed the amplitudes for the two possible polarities.}
	\label{fig:tikzedge}
\end{figure}
The polarity is a necessary and arbitrary choice made at the algorithm's initialization \textit{independently for every edge}. We discuss why it is necessary and a way to make that choice at the end of this section. 
The full unitary evolution of the walk reads~: $\ket{\Phi_{t+1}} = S \times (I\otimes C) \times \ket{\Phi_{t}}$, where $C$ is the local coin operation acting on the coin register~: 
$$
\forall (u,v)\in E, \; \ket{(u,v)}\ket{\pm} \overset{\text{coin}}{\longmapsto} (I\otimes C) \times \ket{(u,v)}\ket{\pm} =  \ket{(u,v)} (C\ket{\pm}).
$$
and $S$ is the scattering which moves the complex amplitudes $\psi_{u,v}^\pm$ according to the coin state. A practical choice for the scattering operator is the Grover diffusion as it is independent of the neighbors' ordering:
$$
\forall u \in V,\; \left(\psi_{u,v}^{\sigma(u,v)}\right)_{v\in V} \overset{\text{scattering}}{\longmapsto} D_{\text{deg}(u)}\times \left(\psi_{u,v}^{\sigma(u,v)}\right)_{v\in V},
$$
where $D_n = \left(\frac{2}{n}\right) _{i,j} - I_n$.

As an example of the above dynamics, one can consider the path of size 3 with the nodes $\{u,v,w\}$, with polarity $\sigma(u,v) = \sigma(v,w) = +$. Figure \ref{fig:ring} shows the unitary evolution of the walker from step $t$ to step $t+1$, when the coin coincides with the first Pauli matrix $X$.

\begin{figure}
	\centering
	\includegraphics[width=\textwidth]{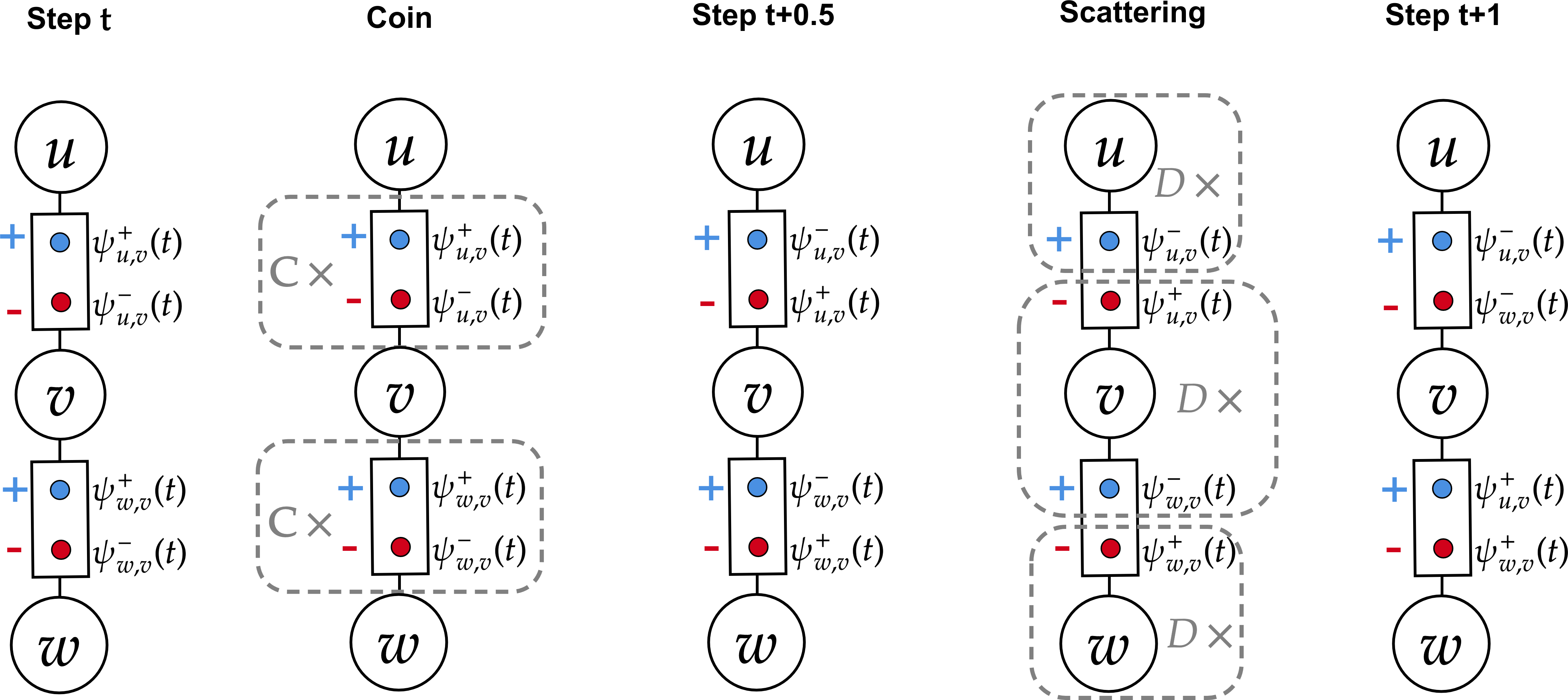}
	\caption{Example of a walk on a path of size 3.}
	\label{fig:ring}
\end{figure}

\paragraph{Polarity : Why and How}
There are two reasons we need polarity. 
First, if the coin operator acts differently on $\ket{+}$ and $\ket{-}$, then different polarities lead to different dynamics. Polarities, in fact, determines how each edge's state scatters within the network, analogously to other spatial searching algorithms on regular lattice~\cite{bezerra2021quantum, arrighi2018dirac, roget2020grover, portugal2013quantum}.
Moreover, polarity is used to divide the edges in two. So each node can only access one amplitude of each edge (depending on the polarity of the edges). This makes things much simpler for a distributed implementation, since we do not have to consider the case of two simultaneous operations on the same amplitude during scattering. 
One way to initialize the polarity on the graph is to create a coloring of $G$ at the initial time step. Afterward, edges have their $+$ pole in the direction of the node with the higher color. This design is especially convenient for bipartite graphs where every node sees only $+$ polarities or only $-$ polarities.

\subsection{Relation with other models}

In this subsection, we would like to discuss the relation between our model of discrete time quantum walk on graphs and other existing models in the literature: the flip-flop quantum walk and the Szegedy quantum walk.

\paragraph{Flip-flop quantum walk}
The flip-flop quantum walk model is the closest from ours. The walker's position is defined onto the nodes of the graph and state $\ket{u}\ket{v}$ symbolizes the walker being on node $u$ and going toward node $v$. The operation around the node that we call in our model scattering is now called the coin step as it acts without changing position of the flip-flop walker. The name flip-flop comes from the operation around the edges that acts by reversing origin and destination ($Q:\ket{u}\ket{v} \mapsto \ket{v}\ket{u}$). The flip-flop quantum walk coin dimension's is variable which makes changing the coin unpractical. Moreover, it faces the same problem we face with edges, namely the ordering of the neighbors which makes the dynamic of the walker for non-regular graphs and a non-Grover coin challenging. In contrast, the coin of our model is of dimension 2 and we provide the notion of polarity, essential to define arbitrary coins. All in all, the flip-flop quantum walk's dynamic is equivalent to our model's when our coin is set to $X$.

\paragraph{Szegedy quantum walk}
The Szegedy quantum walk does not include a coin operator. It is defined from a Markov chain $P$. Let us define the projections
$$
\Pi_A = \sum_x \ket{\alpha_x}\bra{\alpha_x} \qquad \text{and} \qquad \Pi_B = \sum_y \ket{\beta_y}\bra{\beta_y},
$$
where
$$
\ket{\alpha_x} = \ket{x} \otimes \sum_y P_{x,y}\ket{y} \qquad \text{and} \qquad \ket{\beta_y} = \sum_x P_{y,x}\ket{x} \otimes \ket{y}.
$$
For one time step, the unitary evolution reads $U = (2\Pi_B - I)(2\Pi_A - I)$. Now, let's note that if we apply the flip-flop operation $Q$ on $\ket{\beta_y}$ we obtain $\ket{\alpha_y}$. Thus it holds that $U = Q(2\Pi_A - I)Q(2\Pi_A - I)$. Therefore we can consider that $(2\Pi_A - I)$ is a special scattering (for our model) or a special coin (for the flip-flop model) and $Q$ the basic operation around the edges. Thus we get a model really similar to both previous ones, at the difference that two steps are made each iteration. We can also go further and look at the basic Markov chain used, which is $P_{x,y} = \frac{1}{\deg x}$. In such case, one can show that $(2\Pi_A - I)$ is actually the Grover diffusion around every node. This allows us to retrieve the same dynamic as the Grover flip-flop walk or our model with Grover diffusion and coin $X$. The difference between the previous model lies in the what can be changed. Indeed, the Szegedy's model only allow to change the scattering through the Markov chain $P$. However, its properties are strongly linked to the Markov chain $P$ which makes obtaining analytical results much easier.

\paragraph{What about searching ?}
One of the most well known applications of quantum walks is searching. In doing so, one or several positions are marked in some way, often through an oracle like in the Grover algorithm. This leads to a sinusoidal probability to be on a marked position, similar to the Grover algorithm's one. And as in the Grover algorithm, the complexity (composed of the probability of success and hitting time) is going to depend on the actual number of amplitudes marked. In that regard, our model differs from the two others as its walker is located onto the edges. This means that an oracle should mark the edges of the graph, not the nodes. This also means that only two amplitudes per edge are be marked. In contrast, the two other models mark nodes which means $d$ amplitudes are marked. However, our model is searching edges while the two others are searching nodes. It is possible through self loops or other tricks to use our model to search nodes, but as we mark less amplitudes, the performances are going to be lower in the general case. However, the fact that we are marking a constant number of amplitudes mights make studying our model easier for complex topologies. Note that Szegedy quantum walk marks nodes through changing the Markov chain which makes obtaining analytical results much easier.

\paragraph{Conclusion}
In conclusion, the three models are really similar. Their dynamics can be made identical by choosing the right setting. However, each of them has different pros and cons. Our model has a coin of dimension 2 and makes it possible to search edges. The flip-flop model insists more on changing the operations around nodes and does a better job at searching nodes. Finally the Szegedy quantum walk loses a bit of the freedom of the two previous models but allows solid analytical results both for the dynamics and searching.
	
	\section{Distributed implementation}\label{sec:distrib}
In this section we move to the multi-particle states dynamics, allowing a distributed implementation of the above quantum walk. We first introduce the model of distributed computation we use, then we present two distributed protocols to reproduce the dynamic of the walk.

\subsection{Model of distributed quantum computation}
The computation model we consider is a network of qubits, following the topology of a given graph $G=(V,E)$. We place qubits on both edges and nodes. We assume that both nodes and edges can apply one and two qubits quantum gates. As usual, the applications of quantum gates is supervised by classical algorithms. In order to exhibit the communications between the qubits, we sometimes explicit these classical algorithms. The communication cost is calculated by the number of two-qubits gates applied between neighboring qubits. This model is represented in Figure \ref{fig:network}.

\paragraph{Edge register}
Each edge $(u,v)$ has a register of two qubits respectively noted $q_{(u,v)}$ and $q_{(v,u)}$. These two qubits are connected with each other (i.e. we can apply two qubits gate between them). Figure \ref{fig:network} shows the qubits as black dot and the connection between two qubits of one edge in serpentine red lines. Applying a gate between these two qubits is counted as an edge communication.

\paragraph{Node register}
Similarly to the edges, each node hosts qubits. Here we consider two cases according to the connectivity between the qubits inside the nodes. Either the qubits are all-to-all connected or they are connected in a cycle. An all-to-all connectivity implies that all the qubits of a node $u$ are connected between each other and connected to all the edge qubits around node $u$. A cycle connectivity implies that the qubits of a node are connected in a cycle and each node qubit is connected to one edge qubit. These two connectivities are represented in the dotted boxes of Figure \ref{fig:network}. For our goal, which is to use this model to reproduce the dynamic of the quantum walk of Section \ref{sec:model}, using $O(\log d)$ qubits per node for the all-to-all connectivity or $2n$ qubits per node for the cycle connectivity is sufficient.

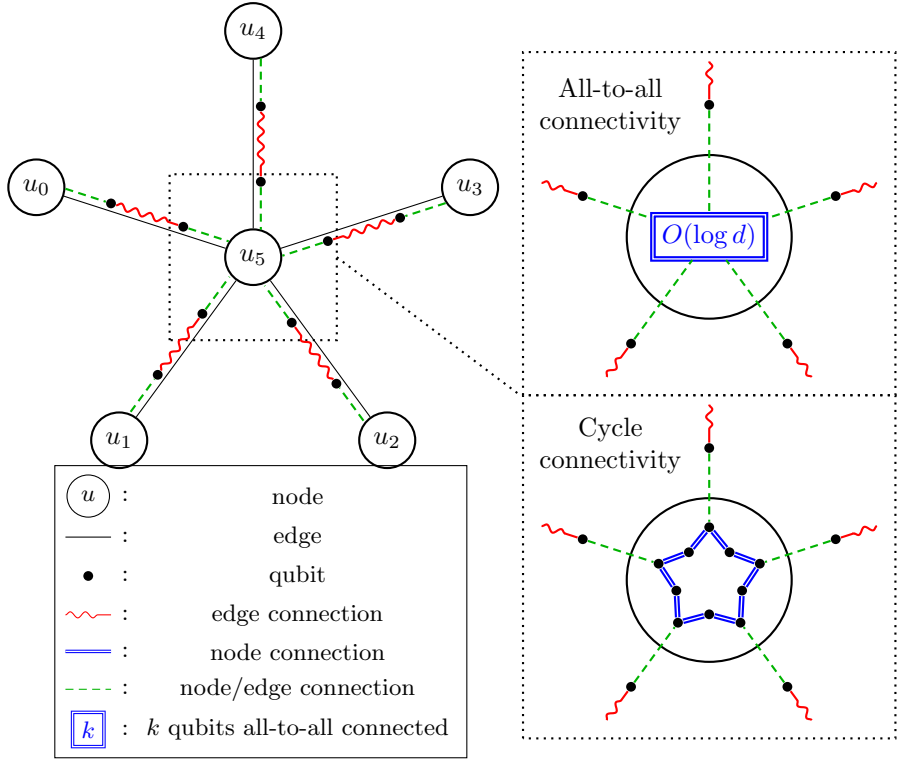
\begin{figure}
	\begin{tikzpicture}[bullet/.style={circle, fill, inner sep=1.3pt}]
		
		\draw[dotted, thick] (1.1,2) -- (3.5,0.18);
		\node at (4.7,4) {$\begin{matrix}\text{All-to-all}\\\text{connectivity}\end{matrix}$};
		\node at (4.7,-0.5) {$\begin{matrix}\text{Cycle}\\\text{connectivity}\end{matrix}$};
		
		\begin{scope}[shift={(0,2)}] 
			\draw[thick, dotted] (-1.1,-1.1) rectangle (1.1,1.1);
			
			\node[draw, circle, thick] (u5) at (0,0) {$u_5$};
			
			\foreach \i in {0,...,4} {
				\begin{scope}[rotate=72*\i-18]
					\node[draw, circle, thick] (u\i) at (-3,0) {$u_\i$};
				\end{scope}
			}
			
			\foreach \i in {0,...,4} {
				\draw (u\i) -- (u5);
				\begin{scope}[rotate=72*\i-18]
					\node[] (a\i) at (-2.76,0.1) {};
					\node[bullet] (b\i) at (-2,0.1) {};
					\node[bullet] (c\i) at (-1,0.1) {};
					\node[] (d\i) at (-0.24,0.1) {};
					\draw[darkgreen, densely dashed, thick] (a\i) -- (b\i);
					\draw[red,decorate,decoration={snake,amplitude=.4mm,segment length=2mm,post length=1mm}, thick] (b\i) -- (c\i);
					\draw[darkgreen, densely dashed, thick] (c\i) -- (d\i);
				\end{scope}
			}
		\end{scope}

		\begin{scope}[shift={(6,-2.28)}, scale=0.7] 
			\draw[thick, dotted] (-3.5,-3) rectangle (3.5,3.5);
			\draw[thick] (0,0) circle [radius=1.55];
			
			\foreach \i in {0,...,4} {
				\begin{scope}[rotate=72*\i-18]
					\def \Pa {(-2.5,0)};
					\def \Pb {(-1,0)};
					\def \Pc {(-0.65,0)};
					\def \Px {(-3.5,0)};
					\pgfmathparse{72*\i-18};\pgfmathparse{int(\pgfmathresult)};\edef \angle {\pgfmathresult};
					\pgfmathparse{3*\i};\pgfmathparse{int(\pgfmathresult)};\edef \a {\pgfmathresult};
					\pgfmathparse{3*\i+1};\pgfmathparse{int(\pgfmathresult)};\edef \b {\pgfmathresult};
					\pgfmathparse{3*\i+2};\pgfmathparse{int(\pgfmathresult)};\edef \c {\pgfmathresult};
					\node[] (x\i) at \Px {};
					\node[bullet] (\a) at \Pa {};
					\node[bullet] (\b) at \Pb {};
					
					
					\begin{scope}[rotate=36]
						\node[bullet] (\c) at (-0.65,0) {};
					\end{scope}
				\end{scope}
			}
			
			\foreach \i in {0,...,4} {
				\pgfmathparse{3*\i};\pgfmathparse{int(\pgfmathresult)};\edef \a {\pgfmathresult};
				\pgfmathparse{3*\i+1};\pgfmathparse{int(\pgfmathresult)};\edef \b {\pgfmathresult};
				\pgfmathparse{3*\i+2};\pgfmathparse{int(\pgfmathresult)};\edef \c {\pgfmathresult};
				\pgfmathparse{3*\i+4};\pgfmathparse{int(\pgfmathresult)};\edef \d {\pgfmathresult};
				\ifthenelse{\i=4}{\def \d {1};}{}
				\draw[red,decorate,decoration={snake,amplitude=.4mm,segment length=2mm,post length=1mm}, thick] (x\i) -- (\a);
				\draw[darkgreen, densely dashed, thick] (\a) -- (\b);
				\draw[blue, double distance=0.5pt, thick] (\b) -- (\c);
				\draw[blue, double distance=0.5pt, thick] (\c) -- (\d);
			}
		\end{scope}

		\begin{scope}[shift={(6,2.27)}, scale=0.7] 
			\draw[thick, dotted] (-3.5,-3) rectangle (3.5,3.5);
			\draw[thick] (0,0) circle [radius=1.55];
			
			\foreach \i in {0,...,4} {
				\begin{scope}[rotate=72*\i-18]
					\def \Pa {(-2.5,0)};
					\def \Pb {(-1,0)};
					\def \Pc {(-0.65,0)};
					\def \Px {(-3.5,0)};
					\pgfmathparse{72*\i-18};\pgfmathparse{int(\pgfmathresult)};\edef \angle {\pgfmathresult};
					\pgfmathparse{3*\i};\pgfmathparse{int(\pgfmathresult)};\edef \a {\pgfmathresult};
					\pgfmathparse{3*\i+1};\pgfmathparse{int(\pgfmathresult)};\edef \b {\pgfmathresult};
					\pgfmathparse{3*\i+2};\pgfmathparse{int(\pgfmathresult)};\edef \c {\pgfmathresult};
					\node[] (x\i) at \Px {};
					\node[bullet] (\a) at \Pa {};
				\end{scope}
			}
			
			\node[draw, rectangle, blue, double distance=0.5pt, thick] (register) at (0,0) {$O(\log d)$};
			
			\foreach \i in {0,...,4} {
				\pgfmathparse{3*\i};\pgfmathparse{int(\pgfmathresult)};\edef \a {\pgfmathresult};
				\pgfmathparse{3*\i+1};\pgfmathparse{int(\pgfmathresult)};\edef \b {\pgfmathresult};
				\pgfmathparse{3*\i+2};\pgfmathparse{int(\pgfmathresult)};\edef \c {\pgfmathresult};
				\pgfmathparse{3*\i+4};\pgfmathparse{int(\pgfmathresult)};\edef \d {\pgfmathresult};
				\ifthenelse{\i=4}{\def \d {1};}{}
				\draw[red,decorate,decoration={snake,amplitude=.4mm,segment length=2mm,post length=1mm}, thick] (x\i) -- (\a);
				\draw[darkgreen, densely dashed, thick] (\a) -- (register);
			}
		\end{scope}

		\begin{scope}[shift={(-0.5,-3)}, scale=0.3] 
			\node [matrix,draw] (my matrix) at (2,1)
			{
				\node[draw, circle] (u) at (0,0) {\scriptsize$u$}; &\node {\scriptsize:}; & \node {\scriptsize node}; \\
				\draw (-0.3,0) -- (0.3,0); &\node {\scriptsize:}; & \node {\scriptsize edge}; \\
				\node[bullet] (q) at (0,0) {}; &\node {\scriptsize:}; & \node {\scriptsize qubit}; \\
				\draw[red,decorate,decoration={snake,amplitude=.4mm,segment length=2mm,post length=1mm}] (-0.3,0) -- (0.3,0); &\node {\scriptsize:}; & \node {\scriptsize edge connection}; \\
				\draw[blue, double distance=0.5pt] (-0.3,0) -- (0.3,0); &\node {\scriptsize:}; & \node {\scriptsize node connection}; \\
				\draw[darkgreen, densely dashed] (-0.3,0) -- (0.3,0); &\node {\scriptsize:}; & \node {\scriptsize node/edge connection}; \\
				\node[draw, rectangle, blue, double distance=0.5pt] (register) at (0,0) {\scriptsize$k$}; &\node {\scriptsize:}; & \node {\scriptsize$k$ qubits all-to-all connected}; \\
			};
		\end{scope}

	\end{tikzpicture}
	\caption{A graphical representation of the model of distributed computation. Show how the qubits are positioned and connected for a given graph. Two node connectivity are considered: all-to-all and cycle.}
	\label{fig:network}
\end{figure}

\subsection{Reproduction of the quantum walk's dynamic}
There are two main points allowing us to reproduce the dynamic of Section \ref{sec:model}' quantum walk. First, we associate each amplitude of the quantum walk to a qubit of the model represented in Figure \ref{fig:network}. Thus the qubit $q_{(u,v)}$ of the distributed model is associated to the amplitude $\ket{u,v}\ket{\sigma(u,v)}$ of the quantum walk. All the other qubits (situated on the nodes) are considered anscillary qubits and will never be measured. The second point regards the states we are allowed to use. Indeed, we restrict ourselves to linear decomposition of unary states. Here, we call unary states a state in which all the qubits are at $\ket{0}$ except exactly one, which is at $\ket{1}$. An example of such unary state is the well known $W$ state defined as 
$$
W = \frac{1}{\sqrt 3}\left(\ket{001} + \ket{010} + \ket{100}\right).
$$

\paragraph{Edges Register}
The edge register represents the position of the walker and is the one measured at the end of the algorithm. It consists of two qubits per edge, corresponding, respectively, to the + and - polarity amplitudes.
For the sake of simplifying the notations, let us arbitrary enumerate the edges of $E$ such that we have $E = \{e_1,\ldots,e_{|E|}\}$. We define $\delta_k^n = \underbrace{0\ldots0}_{k-1 \text{ times}}1\underbrace{0\ldots0}_{n-k \text{ times}}$. The complete state of the walk is a linear combination of all the $\left( \ket{\delta_k^{2|E|}}\right) _{k \in 2|E|}$, where $\ket{\delta_{2k}^{2|E|}}$ corresponds to the amplitude $\ket{e_k}\ket{+}$ of the local model, and $\ket{\delta_{2k+1}^{2|E|}}$ corresponds to amplitude $\ket{e_k}\ket{-}$. It is important that the global state of the register remains a superposition of $\left( \ket{\delta_k^{2|E|}}\right) _{k \in 2|E|}$, as this allows us to obtain a valid solution to the search problem during measurement. In fact, no matter which state $\delta_k^{2|E|}$ is measured, all edges measure state 0 except one, which measures state 1.

\paragraph{Coin}
In the distributed model, the coin operation translates to the application of a $4\times 4$ unitary acting on the subspace composed of the unary states. We call this kind of unitary a unary extension, defined as follows~:
\begin{definition}\label{def:unary_extension}
	Let us have a $n\times n$ unitary $U$. We call $\Lambda(U)$ (of size $2^n\times 2^n$) the unary extension of $U$. $\Lambda(U)$ satisfies the following properties:
	$$
	\forall x \in \{0,\ldots,n-1\}, \quad U\ket{x} = \sum_{i=0}^{n-1} a_i\ket{i} \Rightarrow \Lambda(U)\ket{\delta_x^n} = \sum_{i=0}^{n-1} a_i\ket{\delta_i^n},
	$$
	and
	$$\Lambda(U) \ket{0\ldots 0} = \ket{0\ldots 0}.$$
\end{definition}
Therefore, we want to apply the unary extension of the coin $C$ on the edge qubits of each edge of the distributed model. As shown in Theorem \ref{th:circuit_coin}, there is a circuit of depth $O(1)$ that applies $\Lambda (C)$ on $e_{2i},e_{2i+1}$.

\begin{theorem}\label{th:circuit_coin}
	Let $U$ be a $2\times 2$ unitary operator (a one qubit gate). It holds that
	$$
	\Lambda(U) = \left(\begin{matrix}
		1 & 0 & 0 & 0\\
		0 & U_{00} & U_{01} & 0\\
		0 & U_{10} & U_{11} & 0\\
		0 & 0 & 0 & 1\\
	\end{matrix}\right).
	$$
	Furthermore, $\Lambda(U)$ is realized by the following circuit: \hspace{0.5cm}\Qcircuit @C=1em @R=.7em {
		& \ctrl{1} & \gate{U} & \ctrl{1} & \qw \\
		& \targ & \ctrl{-1} & \targ & \qw
	}

	In this circuit, the CNOT applications are forcing the controlled application of $U$ to be applied only on the basis states $\ket{01}$ and $\ket{10}$.
\end{theorem}

\paragraph{Scattering}
Reproducing the scattering operation, like the coin operation, requires applying the unary extension of the scattering operator. This unary extension must be applied locally around the nodes. The polarity ensures that these local scatterings commute and can be applied simultaneously. Hence, the distributed scheme for the scattering around one node is defined as follows.

\begin{definition}{Distributed scattering around a node}\label{def:scatter}
	Let $v$ be a node of degree $d$. Furthermore, let us enumerate the quantum walk amplitudes around to $v$ (as defined by the polarity) $a_0,\ldots,a_
	{d-1}$ with the map
	$$
	\phi : \left\{\begin{matrix}
		2k &\mapsto& \ket{e_k}\ket{+}\\
		2k+1 &\mapsto& \ket{e_k}\ket{-}\\
	\end{matrix}\right.
	$$
	This numbering can be imposed by the hardware or arbitrarily set and is necessary to define the scattering operator $U(d)$. Note that the Grover diffusion operator is independent of the numbering.
	In the quantum walk dynamic of Section \ref{sec:model}, we apply $U(d)$ on $\text{Vec}\left(\left\{ \phi(a_i)\right\}_{0\leq k< d} \right)$.
	In order to reproduce this dynamic with the distributed model, we apply $\Lambda(U)$ on the qubits $a_0,\ldots,a_
	{d-1}$.
\end{definition}

In the next sections, we present two ways to achieve that goal, one for each connectivity considered in the distributed model (see Figure \ref{fig:network}).

\subsection{Scattering operation with all-to-all connectivity}
In this section, we consider the distributed model with all-to-all connectivity. We show how to implement the scattering operation of the quantum walk dynamic presented in Section \ref{sec:model}. 

According to Definition \ref{def:scatter}, we have an edge register $\mathcal{E}$ containing $d$ qubits $a_0,\ldots,a_{d-1}$ where $d$ is the degree of the node. In order to apply the scattering, we introduce two ancillary registers $\mathcal N_1$ and $\mathcal N_2$ on the node. $\mathcal N_1$ contains $\lceil \log d \rceil$ qubits and while $\mathcal N_2$ has one. We now define an operation $T$ that transfers the unary state on register $\mathcal E$ to the register $\mathcal N_1$ in a binary format while marking non-zero state transfers in register $\mathcal N_2$. Formally,
$$
\forall k, \quad T : \left\{\begin{matrix}
	\ket{\delta_k^d}_{\mathcal{E}} \ket{0}_{\mathcal{N}_1} \ket{0}_{\mathcal{N}_2} &\mapsto& \ket{0}_{\mathcal{E}} \ket{k}_{\mathcal{N}_1} \ket{1}_{\mathcal{N}_2}\\
	\ket{0}_{\mathcal{E}} \ket{0}_{\mathcal{N}_1} \ket{0}_{\mathcal{N}_2} &\mapsto& \ket{0}_{\mathcal{E}} \ket{0}_{\mathcal{N}_1} \ket{0}_{\mathcal{N}_2}\\
\end{matrix}\right. .
$$

Let us denote the application of $U(d)$ on register $\mathcal N_1$ controlled by register $\mathcal N_2$ as $C_{\mathcal N_2}[U(d)_{\mathcal N_1}]$. One can show that 
$$
\Lambda(U)_{\mathcal E} \otimes I_{\mathcal N_1} \otimes I_{\mathcal N_2}  = T^{-1} \times (I_{\mathcal E} \otimes C_{\mathcal N_2}[U(d)_{\mathcal N_1}])\times T.
$$

It is actually possible to decompose $T$ into a product of $d$ commutative operators $T_k$, each acting on the registers $\mathcal N_1$ and $\mathcal N_2$ as well as qubit $a_k$.
$$
\forall k, \quad T_k : \left\{\begin{matrix}
	\ket{1}_{a_k} \ket{0}_{\mathcal{N}_1} \ket{0}_{\mathcal{N}_2} &\mapsto& \ket{0}_{a_k} \ket{k}_{\mathcal{N}_1} \ket{1}_{\mathcal{N}_2}\\
	\ket{0}_{a_k} \ket{0}_{\mathcal{N}_1} \ket{0}_{\mathcal{N}_2} &\mapsto& \ket{0}_{a_k} \ket{0}_{\mathcal{N}_1} \ket{0}_{\mathcal{N}_2}\\
\end{matrix}\right. .
$$

It holds that 
$$
T = \prod_k T_k.
$$

In Algorithm \ref{algo:Trk}, we provide a distributed scheme for $T_k$. It is running on the node and it needs two communication methods: \texttt{RequestCnot}(edge, target), which applies a Not on target controlled by the qubit in edge accessible, according to polarity; and \texttt{ApplyMCT}(edge) which applies a Not on the edge's qubit controlled by the full node register $\mathcal N_1 \cup \mathcal N_2$. Note that, while \texttt{ApplyMCT}(edge) has significant computational cost, it is possible to limit the communication cost to one CNOT between the edge register and the node register by using an anscillary qubit as seen in Figure \ref{fig:mct_application}. The MCT (Multi-Controlled Tofolli gate) gate can then be applied on the node register with depth $O(d)$, as shown in Biswal et al \cite{biswal2019techniques}.

\begin{figure}
	\[\Qcircuit @C=1em @R=.7em {
		\lstick{q_0}	& \ctrl{2} & \qw & \ctrl{2} & \qw \\
		\lstick{q_1}	& \ctrl{1} & \qw & \ctrl{1} & \qw \\
		\lstick{q_2}	& \targ & \ctrl{1} & \targ & \qw \\
		\lstick{a_i}	& \qw & \targ & \qw & \qw
	}\qquad
	\begin{matrix}
		\\
		\\
		\\
		\equiv\\
	\end{matrix}
	\qquad\quad
	\Qcircuit @C=1em @R=.7em {
		\lstick{q_0}	& \ctrl{3} & \qw \\
		\lstick{q_1}	& \ctrl{2} &  \qw \\
		\lstick{q_2}	& \qw  & \qw \\
		\lstick{a_i}	& \targ & \qw
	}
	\]
	\caption{Application of a MCT (multi-controlled X gate) on an edge controlled by the whole node register via a CNOT gate and an anscillary qubit. $Q_0,q_1$ are the qubits of the node register, $q_2$ is the anscillary qubit and $a_i$ is the targeted edge qubit. Both circuits are equivalent.}
	\label{fig:mct_application}
\end{figure}
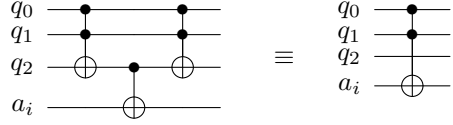

Algorithm \ref{algo:Trk} uses $1+\log d$ calls to \texttt{RequestCnot}(edge, target) and one call to \texttt{ApplyMCT}(edge) which results in a communication cost of $2+\log d$ CNOT between the node and an edge. We need to apply it $2d$ times for one scattering. Thus one scattering admits a communication cost of $2d(2+\log d) = O(d\log d)$ CNOT.

Finally, Figure \ref{fig:circuit} shows an example of this distributed design of the path graph of five nodes. The path graph has the particularity of having the same topology as a circuit (a qubit being connected to the preceding and following qubit). Notice that such circuit coincides with a partitioned QCA, each operation is local and translational invariant. 

\begin{algorithm}[h]
	\caption{Distributed scheme for $\text{T}_k$ on node $u$ of degree $d$}
	\label{algo:Trk}
	\begin{algorithmic}[1]
		\Require $e$ an edge connected to $u$
		\Require $1\leq k \leq d$
		\Function{$\text{T}_k$}{$e,k$}
		\State $r \gets \lceil\log d\rceil$ \Comment{The size of $u$'s register is $r+1$}
		\For{$0\leq i < r \mid \left((k-1)^{(2)} \right)_i = 1$}
		\State \Call{RequestCnot}{$e, q_i$} \Comment{CNOT on the $i^{\text{th}}$qubit of $u$}
		\EndFor
		
		\State \Call{RequestCnot}{$e, q_r$}
		
		\State Apply $X$ gate on all qubits $q_i$ of $u$ such that $\left((k-1)^{(2)} \right)_i = 0$.
		
		\State \Call{ApplyMCT}{e}
		
		\State Apply $X$ gate on all qubits $q_i$ of $u$ such that $\left((k-1)^{(2)} \right)_i = 0$.
		
		\EndFunction
	\end{algorithmic}
\end{algorithm}

\begin{figure}[h]
	\centering
	\includegraphics[width=\textwidth]{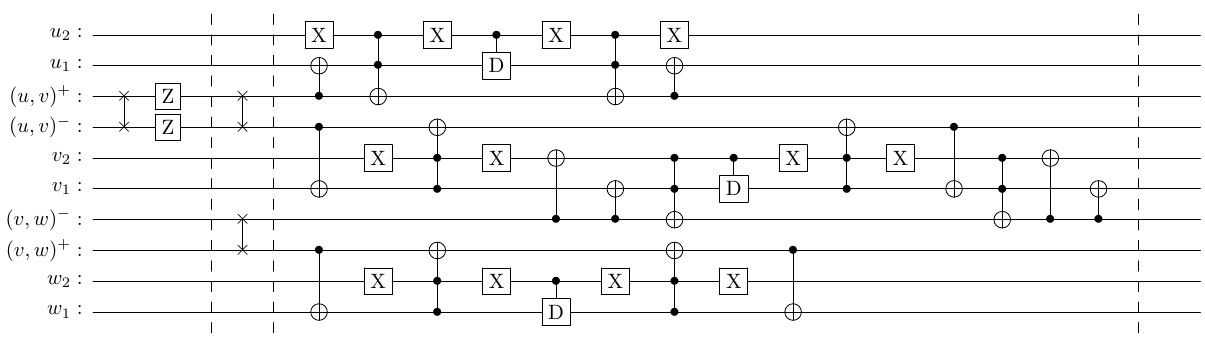}
	\caption{Circuit of one step of the quantum walk for the path graph $u-v-w$. The circuit applies successively the oracle on $(u,v)$, the coin, $T$, $D$, $T^{-1}$.}
	\label{fig:circuit}
\end{figure}

\subsection{QCA-like general scattering}
In this section, we consider the distributed model with cycle connectivity as seen in Figure \ref{fig:network}. Under this new connectivity, we show how to implement the scattering operation of the quantum walk dynamic presented in Section \ref{sec:model}. The main idea here is to define a QCA-like dynamical system onto the node register that can implement any unitary. It is based on three following ideas: 
\begin{enumerate}
	\item Any unitary can be decomposed into a product of two-levels unitaries.
	\item There exists a quantum walk on the cycle with space-time dependent coin which implements any given two-levels unitaries.
	\item The dynamic of a quantum walk on a cycle can be reproduced by a QCA (Quantum Cellular Automata). The resulting unitary is exactly the unary extension of the quantum walk's unitary.
\end{enumerate} 
Therefore, to implement the scattering operation within our distributed model, first we decompose the scattering operator $U$ into two-levels unitaries. Then we calculate the coefficients of the coin operators implementing the two-levels unitaries. And finally, we apply the dynamic of the quantum walk on a cycle with the coin operators, previously calculated via a QCA.

In the next three sections, we detail these ideas. Section \ref{sec:decompUnit} introduces the notion of two-levels unitaries and decomposition. The Theorem \ref{th:decompose_batch}, in particular, show that an application of the two-levels unitaries per batch is possible, optimizing the application of general unitaries. Section \ref{sec:2lvlqw} presents in Theorem \ref{th:transpile} a quantum walk on a cycle that is able to implement any two-levels unitary. Algorithm \ref{algo:transpile} gives the procedure to apply any unitary via a quantum walk on a cycle in $O(d^2)$ steps (where $d$ is the dimension of the unitary). Section \ref{sec:scatqca} shows how the previous quantum walk dynamic can be implemented with a QCA-like protocol that respects the topology represented in Figure \ref{fig:network}. Figure \ref{fig:circuit_scatter}, in particular, the resulting circuit.

\subsubsection{Decomposition of unitaries}\label{sec:decompUnit}
In this section, we show how an unitary operator can be decomposed into two-levels unitaries,as in \cite{nielsen2010quantum}. Two-levels unitaries are unitaries that only acts on two amplitudes. Similarly, we can define a special type of two-levels unitariy that only act on one amplitude as one-level unitary. Definition \ref{def:one-level} and \ref{def:two-level} give a more formal definition of these objects and introduce notations to describe them.
\begin{definition}{One-level unitary operator}\label{def:one-level}
	A one-level unitary operator is the tuple $[c, k, n]$, where $n\in\mathbb{N}, 0 \leq k < n$ and $c \in \mathbb{C}$.
	It denotes the following unitary operation:
	$$
	[c, k, n] = c\ket{k}\bra{k} + \sum_{\begin{matrix}i=0\\i\neq k\end{matrix}}^{n-1} \ket{i}\bra{i}.
	$$
\end{definition}
\begin{definition}{Two-levels unitary operator}\label{def:two-level}
	A two-level unitary operator is the tuple $[U, k_0, k_1, n]$, where $n\in\mathbb{N}, 0 \leq k_0\neq k_1 < n$ and $U$ is a $2\times 2$ unitary.
	It denotes the following unitary operation:
	$$
	[U,k_0,k_1,n] = \sum_{i,j = 0}^1 U_{i,j}\ket{k_i}\bra{k_j} + \sum_{\begin{matrix}i=0\\i\notin \{k_0,k_1\}\end{matrix}}^{n-1} \ket{i}\bra{i}.
	$$
\end{definition}

The main tool at our disposal to decompose unitaries in two-levels unitaries is Theorem \ref{th:decompose_aux}. This result allows us to find a two-levels unitary $M$ such that multiplying together $M$ and $U$ makes one element of $U$ under the diagonal vanish. If the process is repeated in the right order we can find two-levels unitaries $M_1,\ldots,M_k$ such that $M_1\times U$ has one coordinate under the diagonal equal to $0$, $M_2\times M_1 \times U$ has two, and so on and so far. Thus, $M_{k} \times M_{k-1} \times \ldots \times M_1 \times U$ is upper triangular. Since the result of this product is both unitary and upper triangular, it is diagonal. Thus, we apply one-level unitaries $M'1,\ldots,M'_l$ in order to make all the diagonal coefficients equal to one. Finally, we have
$$
M'_l \times M'_{l-1} \times \ldots \times M'_{1} \times M_{k} \times M_{k-1} \times \ldots \times M_1 \times U = I_2.
$$
And we get the resulting decomposition
$$
M_1 \times M_{2} \times \ldots \times M_k \times M'_1 \times \ldots \times M'_l = U.
$$

This decomposition is detailed in Algorithm \ref{algo:decompose_naive}.
\begin{theorem}\label{th:decompose_aux}
	Let $U$ be a $n\times n$ unitary and let $0 \leq k_0 < k_1 < n$. It exists a two-levels unitary $[T_U, k_0,k_1,n]$ such that
	$$
	\forall 0 \leq i,j < n, \quad \{i,j\}\cap\{k_1,k_1\} = \emptyset \Rightarrow\left([T_U,k_0,k_1,n] \times U\right)_{i,j} = U_{i,j},
	$$
	and
	$$
	\left([T_U,k_0,k_1,n] \times U\right)_{k_1,k_0} = 0.
	$$
\end{theorem}

\begin{algorithm}
	\caption{Two-level decomposition \cite{nielsen2010quantum}}\label{algo:decompose_naive}
	\begin{algorithmic}
		\Require $U$
		\Ensure $n = \dim U$
		\State $U_1 \gets U$
		\State $k \gets 1$
		\For{$i$ in $0,\ldots,n-1$}
		\For{$j$ in $i+1,\ldots,n-1$}
		\State $M'_k \gets [T_{U_k}, i, j, n]$ \Comment{According to Theorem \ref{th:decompose_aux}}
		\State $M \gets (M')^\dagger = [T_{U_k}^\dagger, i, j, n]$
		\State $U_{k+1} \gets M'_k \times U_k$
		\State $k \gets k+1$
		\EndFor
		\EndFor
		\For{$i$ in $0,\ldots,n-1$}
		\State $M_k \gets [(U_k)_{i,i}, i, n]$ \Comment{one-level unitary}
		\State $U_{k+1} \gets U_k$
		\State $k \gets k+1$
		\EndFor
	\end{algorithmic}
\end{algorithm}

We now present a way to apply these one-level and two-level unitaries per batch, where all the unitaries of one batch commuting with each others. This reduces the computational cost of implementing arbitrary unitaries since we can apply several two-levels unitaries at the same time.
\begin{theorem}[Batched application]\label{th:decompose_batch}
	For all unitary $U$, there exists $B_1,\ldots,B_k$ sets of one-level and two-level unitaries such that 
	$$
	U = \oprod{i=1}{k} \prod_{M\in B_k}M,
	$$
	where $k = O (\dim U)$.
	
	The one-level and two level unitaries of a same batch commute.
\end{theorem}
\begin{proof}
	Using Theorem \ref{th:decompose_aux}, we eliminate all the non-diagonal coordinates of $U$ one by one, essentially computing $U^\dagger$. The condition to that is to compute the two-level unitaries in the right order. Indeed, as long as the coordinates left and top of target coordinate $(k_1,k_0)$ of Theorem \ref{th:decompose_aux} are zero, then they will remain zero after application of a two-level unitary of the form $[M,k_0,k_1,n]$. Furthermore, two two-level unitaries $[M_1,k_0,k_1,n]$ and $[M_1,l_0,l_1,n]$ commutes if $\{k_0,k_1\}\cap \{l_1,l_1\} = \emptyset$ as there non-trivial operations are on different subspaces. Hence, we can organize the batches of target coordinates to eliminate at the same time all coordinates on the same off-antidiagonals (off-diagonals parallel to the antidiagonal). One last batch of one-level unitaries needs to be computed in order to finish the decomposition. As there is $2n-3$ off-antidiagonals, there is $2n-2 = O(n)$ batches. We provide a pseudo-code of this new scheme in Algorithm \ref{algo:decompose_batch}.
	\begin{algorithm}
		\caption{Batched two-level decomposition}\label{algo:decompose_batch}
		\begin{algorithmic}
			\Require $U$
			\Ensure $n = \dim U$
			\State $U_1 \gets U$
			\State $k \gets 1$
			\For{$s$ in $0,\ldots,2n$}
			\State $B_k \gets \emptyset$
			\For{$i$ in $0,\ldots,n-1$}
			\For{$j$ in $i+1,\ldots,n-1$}
			\State $M' \gets [T_{U_k}, i, j, n]$ \Comment{According to Theroem \ref{th:decompose_aux}}
			\State $M \gets (M')^\dagger = [T_{U_k}^\dagger, i, j, n]$
			\State $B_k \gets B_k \cup \{M\}$
			\EndFor
			\EndFor
			\State $U_{k+1} \gets \big(\prod_{M\in B_k}M\big) \times U_k$
			\State $k \gets k+1$
			\EndFor
			\State $B_k \gets \emptyset$
			\For{$i$ in $0,\ldots,n-1$}
			\State $M \gets [(U_k)_{i,i}, i, n]$ \Comment{one-level unitary}
			\State $B_k \gets B_k \cup \{M\}$
			\EndFor
		\end{algorithmic}
	\end{algorithm}
\end{proof}

\subsubsection{Application of a two-levels unitary with a cycle quantum walk}\label{sec:2lvlqw}
Now we can show that the one-level and two-level unitaries may be rephrased in terms of a cyclic quantum walk. We define the quantum walk on the $n$ sized cycle with two Hilbert space: one for the position and one for the coin. Thus, the walker lives in  $\mathcal H_n \otimes \mathcal H_2$. In order to fully define a quantum walk, apart from the space, we also need to define the coin and scattering operations. Here the scattering moves the walker to the right when the coin state is $\ket{1}$ and keep the walker in-place otherwise. As for the coin, we make it space-time dependent. More formally,
\begin{description}
	\item[State:] $\displaystyle \ket{\psi_t} \in \mathcal H_n \otimes \mathcal H_2$,
	\item[Coin:] $\displaystyle C_t = \sum_{i=0}^{n-1}\ket{i}\bra{i} \otimes C_t(i)$,
	\item[Scattering:] $\displaystyle S = \sum_{i=0}^{n-1} \ket{i,0}\bra{i,0} + \ket{i+1,1}\bra{i,1}$,
	\item[Update rule at step $t$:] $\displaystyle \mathcal{QW}_t = S \times C_t$,
	\item[Dynamic from step $t_0$ to step $t_1$:] $\displaystyle \mathcal{QW}_{(t_0,t_1)} = \cprod{t=t_0}{t_1} \mathcal{QW}_t = \mathcal{QW}_{t_1} \times \ldots \times \mathcal{QW}_{t_0}$.
\end{description}

Our goal is, given a target unitary $T$ of dimension $2n$, to chose a time $t_f$ and the coins $\big(C_t(x)\big)_{\begin{matrix}0\leq t < t_f\\0\leq x < n\\\end{matrix}}$ such that $\mathcal{QW}_{(0,t_f-1)} = T$. Theorem \ref{th:transpile} shows how to do such a thing for a batch of two-level unitaries while Theorem \ref{th:transpile_diag} shows how to do it for diagonal matrices. 

\begin{theorem} \label{th:transpile}
	Let there be a set of two-level unitaries $\big\{T_k = [M_k,2x_{2k}+v_{2k},2x_{2k+1}+v_{2k+1},2n]\big\}_k$ with all $X = \{x_i\}_i$ different pairwise. For all $k$ and $0\leq t < 2n$, let 
	$$
	C_t(x_{2k}) = \left\{\begin{matrix}
		X & \text{if }\; (t\in\{0,n\}) \wedge (v_{2k}=1)\\
		M_k & \text{if }\; t=t_k \\
		I & \text{otherwise}
	\end{matrix}\right. ,
	$$
	$$
	C_t(x_{2k+1}) = \left\{\begin{matrix}
		X & \text{if }\; (t\in\{0,n\}) \wedge (v_{2k+1}=0)\\
		I & \text{otherwise}
	\end{matrix}\right. ,
	$$
	where 
	$$t_k = \left\{\begin{matrix}
		(x_{2k}-x_{2k+1}+n) \text{ mod } n & \text{if } x_{2k+1}\neq x_{2k}\\
		n & \text{if } x_{2k+1}=x_{2k}
	\end{matrix}\right. .
	$$
	It holds that $\mathcal{QW}_{(0,2n-1)} = \prod_{k}T_k$.
\end{theorem}
\begin{proof}
	We apply $\mathcal{QW}_{(0,2n-1)}$ an all basis vectors and check the equality. There is three cases: $\ket{y}\ket{v}$, $\ket{x_{2k}}\ket{v}$, $\ket{x_{2k+1}}\ket{v}$ where $y\notin X$ and $v\in\{0,1\}$.
	\paragraph{Case 1:}
	\[
	\begin{split}
		\mathcal{QW}_{(0,2n-1)} \ket{y,v}&= \left(\cprod{t=0}{2n-1} S\times C_t\right) \times \ket{y,v}\\
		&= \left(\cprod{t=0}{2n-1} S\times I\right) \ket{y,v}\\
		&=  S^{2n} \ket{y,v}\\
		&=  \ket{y,v}\\
		&= \left(\prod_{i}T_i\right)\ket{y,v}.
	\end{split}
	\]
	
	\paragraph{Case 2:}
	\[
	\begin{split}
		\mathcal{QW}_{(0,2n-1)} \ket{x_{2k},v}&= \left(\cprod{t=0}{2n-1} S\times C_t\right) \times \ket{x_{2k},v}\\
		&= \left(\cprod{t=1}{2n-1} S\times C_t\right) \times S \times \ket{x_{2k},0}\\
		&= \left(\cprod{t=t_k}{2n-1} S\times C_t\right) \times \ket{x_{2k},0}\\
		&= \left(\cprod{t=t_k+1}{2n-1} S\times C_t\right)\times S \times M_k\ket{x_{2k},0}\\
		&= \left(\cprod{t=n}{2n-1} S\times C_t\right) \times M_k\ket{x_{2k},0}\\
		&= \left(\cprod{t=n+1}{2n-1} S\times I\right) \times S \times M_k\ket{x_{2k},v}\\
		&= S^n \times M_k\ket{x_{2k},v}\\
		&= M_k\ket{x_{2k},v}\\
		&= \left(\prod_{i}T_i\right) \ket{x_{2k},v}.
	\end{split}
	\]
	\newpage
	
	\paragraph{Case 3:}
	\[
	\begin{split}
		\mathcal{QW}_{(0,2n-1)} \ket{x_{2k+1},v}&= \left(\cprod{t=0}{2n-1} S\times C_t\right) \times \ket{x_{2k+1},v}\\
		&= \left(\cprod{t=1}{2n-1} S\times C_t\right) \times S \times \ket{x_{2k+1},1}\\
		&= \left(\cprod{t=t_k}{2n-1} S\times C_t\right) \times \ket{x_{2k+1}+t_k,1}\\
		&= \left(\cprod{t=t_k}{2n-1} S\times C_t\right) \times \ket{x_{2k},1}\\
		&= \left(\cprod{t=t_k+1}{2n-1} S\times C_t\right)\times S \times M_k\ket{x_{2k},1}\\
		&= \left(\cprod{t=n}{2n-1} S\times C_t\right) \times M_k\ket{x_{2k}+n-tk,1}\\
		&= \left(\cprod{t=n}{2n-1} S\times C_t\right) \times M_k\ket{x_{2k+1},1}\\
		&= \left(\cprod{t=n+1}{2n-1} S\times I\right) \times S \times M_k\ket{x_{2k+1},v}\\
		&= S^n \times M_k\ket{x_{2k+1},v}\\
		&= M_k\ket{x_{2k+1},v}\\
		&= \left(\prod_{i}T_i\right) \ket{x_{2k+1},v}.
	\end{split}
	\]

\end{proof}

\begin{theorem} \label{th:transpile_diag}
	Let there be a diagonal unitary $D$ of size $2n\times 2n$.
	For all $0\leq x < n$ and $0\leq t < n$, let 
	$$
	C_t(x) = \left\{\begin{matrix}
		\left(\begin{matrix}
			D_{2x} & 0\\
			0 & D_{2x+1}\\
		\end{matrix}\right) & \text{if }\; (t=0) \\
		I & \text{otherwise}
	\end{matrix}\right. .
	$$
	It holds that $\mathcal{QW}_{(0,2n-1)} = D$.
\end{theorem}
\begin{proof}
	First note that 
	$$
	C_0 = \sum_{x=0}^{n-1} \ket{x}\bra{x} \otimes \left(\begin{matrix}D_{2x} & 0\\0 & D_{2x+1}\\\end{matrix}\right) = D.
	$$
	
	It holds that
	\[
	\begin{split}
		\mathcal{QW}_{(0,2n-1)} &= \cprod{t=0}{2n-1} \mathcal{QW}_t\\
		&= \cprod{t=0}{2n-1} (S \times C_t)\\
		&= \left(\cprod{t=1}{2n-1} (S \times C_t)\right) \times S \times C_0\\
		&= \left(\cprod{t=1}{2n-1} S \right) \times S \times D\\
		&= \left(\cprod{t=0}{2n-1} S \right) \times D\\
		&=  D\\
	\end{split}
	\]
\end{proof}

Combining those two results with Algorithm \ref{algo:decompose_batch}, we can, for a target unitary $T$, compute the coins $C_t(x)$ such that $U_{(0,t_f-1)} = T$, where $t_f = \underbrace{2n}_{\text{cost of a batch}}\times\underbrace{(8n-5)}_{\text{number of batches}}$. The way to compute those coins is described in Algorithm \ref{algo:transpile}.

\begin{algorithm}
	\caption{Batched two-level application}\label{algo:transpile}
	\begin{algorithmic}
		\Require $U$
		\Ensure $n = \dim U$
		\State $U_1 \gets U$
		\State $k \gets 1$
		\For{$s$ in $0,\ldots,2n$} \Comment{Lists the off-antidiagonals}
		\For{$p$ in $0,1$} \Comment{Ensures that all unitaries of the same batch acts on different position}
		\State $B_k \gets \emptyset$
		\For{$i$ in $0,\ldots,n-1$}
		\For{$j$ in $i+1,\ldots,n-1$}
		\State $M' \gets [T_{U_k}, i, j, n]$ \Comment{According to Theroem \ref{th:decompose_aux}}
		\State $M \gets (M')^\dagger = [T_{U_k}^\dagger, i, j, n]$
		\State $B_k \gets B_k \cup \{M\}$
		\EndFor
		\EndFor
		\State $U_{k+1} \gets \big(\prod_{M\in B_k}M\big) \times U_k$
		\State $k \gets k+1$
		\EndFor
		\EndFor
		\State $\big(C_t(x)\big)_{0\leq t < 2n-1} \gets \textsc{Transpile}(U_k)$ \Comment{Compute $C_t(x)$ for target $U_k$ according to Theorem \ref{th:transpile_diag}}
		\For{$i$ in $k-1,\ldots,0$}
		\State $\big(C_t(x)\big)_{2nk\leq t < 2n(k+1)-1} \gets \textsc{Transpile}(U_i)$ \Comment{Compute $C_t(x)$ for target $U_i$ according to Theorem \ref{th:transpile}}
		\EndFor
	\end{algorithmic}
\end{algorithm}

\subsubsection{Scattering on the cycle}\label{sec:scatqca}
While previous section shows how to apply an unitary via a quantum walk, it still does not solve the original problem of applying a distributed scattering. In order to do that, we provide a distributed implementation of the previous cycle quantum walk. We recall that according to Definition \ref{def:scatter}, we want to apply $\Lambda(U)$ on the qubits $a_0,\ldots,a_
{d-1}$, for a scattering operator $U$.

Let $\big(C_t(x)\big)_{\begin{matrix}0\leq t < t_f\\0\leq x < d\\\end{matrix}}$ be the coins of the cycle quantum walk $\mathcal{QW}$ of size $d$ such that 
$$
\displaystyle \mathcal{QW}_{(0,t_f-1)} = \cprod{t=0}{t_f-1} \big(S \times \sum_{i=0}^{n-1}\ket{i}\bra{i} \otimes C_t(i)\big) =T,
$$
where 
$$
T = U \otimes \ket{0}\bra{0} + I_d \otimes \ket{1}\bra{1} \quad \text{and} \quad t_f = 
$$

We also introduce the anscillary qubits $b_0,\ldots,b_{d-1}$ and $c_0,\ldots,c_{d-1}$. Those are the node register. We suppose the following connectivity: 
$$
\forall i, \; a_i \sim b_i \sim c_i \sim b_{i+1}, \quad \text{where } \sim \text{ means "is connected to"}.
$$

This connectivity, as well as the position of the qubits, is illustrated in Figure \ref{fig:node_register}.

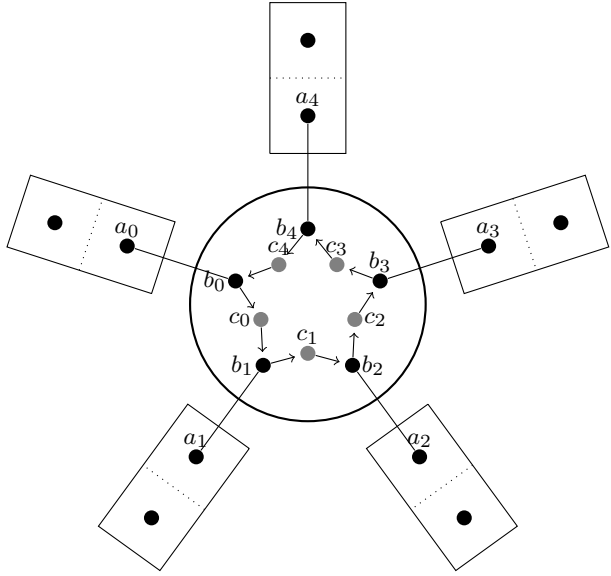
\begin{figure}
	\centering
	\begin{tikzpicture}[bullet/.style={circle, fill, inner sep=2pt}]
		\draw[thick] (0,0) circle [radius=1.55];
		\foreach \i in {0,...,4} {
			\begin{scope}[rotate=72*\i-18]
				\draw (-2,-0.5) -- (-2,0.5) -- (-4,0.5) -- (-4,-0.5) -- (-2,-0.5);
				\draw[dotted] (-3,-0.5) -- (-3,0.5);
				\def \Pa {(-2.5,0)};
				\def \Pb {(-1,0)};
				\def \Pc {(-0.65,0)};
				\def \Px {(-3.5,0)};
				\pgfmathparse{72*\i-18};\pgfmathparse{int(\pgfmathresult)};\edef \angle {\pgfmathresult};
				\pgfmathparse{3*\i};\pgfmathparse{int(\pgfmathresult)};\edef \a {\pgfmathresult};
				\pgfmathparse{3*\i+1};\pgfmathparse{int(\pgfmathresult)};\edef \b {\pgfmathresult};
				\pgfmathparse{3*\i+2};\pgfmathparse{int(\pgfmathresult)};\edef \c {\pgfmathresult};
				\node[bullet] (x\i) at \Px {};
				\node[bullet] (\a) at \Pa {};
				\node[above] at \Pa {\scriptsize{$a_\i$}};
				\node[bullet] (\b) at \Pb {};
				\ifthenelse{\angle < 90 \AND \angle >-19}{\node[left] at \Pb {\scriptsize{$b_\i$}};}{
					\ifthenelse{\angle < 144 \AND \angle >54}{\node[right] at \Pb {\scriptsize{$b_\i$}};}{
						\ifthenelse{\angle < 220 \AND \angle >144}{\node[above] at \Pb {\scriptsize{$b_\i$}};}{
							\ifthenelse{\angle < 360 \AND \angle >280}{\node[left] at \Pb {\scriptsize{$b_\i$}};}{
								\ifthenelse{\angle < 280 \AND \angle >260}{\node[left] at \Pb {\scriptsize{$b_\i$}};}{
									\node[above] at \Pb {\scriptsize{$b_\i$}};}}}}}
				\begin{scope}[rotate=36]
					\node[bullet,gray] (\c) at (-0.65,0) {};
					\ifthenelse{\angle < 54 \AND \angle >-19}{\node[left] at \Pc {\scriptsize{$c_\i$}};}{
						\ifthenelse{\angle < 144 \AND \angle >54}{\node[right] at \Pc {\scriptsize{$c_\i$}};}{
							\ifthenelse{\angle < 234 \AND \angle >144}{\node[above] at \Pc {\scriptsize{$c_\i$}};}{
								\ifthenelse{\angle < 360 \AND \angle >280}{\node[left] at \Pc {\scriptsize{$c_\i$}};}{
									\node[above] at \Pc {\scriptsize{$c_\i$}};}}}}
				\end{scope}
			\end{scope}
		}
		\foreach \i in {0,...,4} {
			\pgfmathparse{3*\i};\pgfmathparse{int(\pgfmathresult)};\edef \a {\pgfmathresult};
			\pgfmathparse{3*\i+1};\pgfmathparse{int(\pgfmathresult)};\edef \b {\pgfmathresult};
			\pgfmathparse{3*\i+2};\pgfmathparse{int(\pgfmathresult)};\edef \c {\pgfmathresult};
			\pgfmathparse{3*\i+4};\pgfmathparse{int(\pgfmathresult)};\edef \d {\pgfmathresult};
			\ifthenelse{\i=4}{\def \d {1};}{}
			\draw (\a) -- (\b);
			\draw[->,shorten >=2pt] (\b) -- (\c);
			\draw[->,shorten >=2pt] (\c) -- (\d);
		}
	\end{tikzpicture}
	\caption{Connectivity of the qubits around a node.}
	\label{fig:node_register}
\end{figure}

We want to apply $\Gamma(T)$ on $a_0,b_0,a_1,b_1,\ldots,a_{d-1},b_{d-1}$. 

We note $M_{[q_1,\ldots,q_n]}$ the unitary operator $M$ applied on qubits $q_1,\ldots,q_n$.

It holds that 
$$
\Lambda(C_t)_{[a_0,b_0,\ldots,a_{d-1},b_{d-1}]} \otimes I_{[c_0,\ldots,c_{d-1}]} = \left(\bigotimes_{i=0}^{d-1} \underbrace{\Lambda(C_t(i))_{[a_i,b_i]}}_{\text{applied on qubits $a_i$ and $b_i$}}\right) \otimes I_{[c_0,\ldots,c_{d-1}]}.
$$
As $C_t(i)$ are all $2\times 2$ unitary and according to Theroem \ref{th:circuit_coin}, $\Lambda(C_t(i))$ admits a circuit with depth $O(1)$. The two-qubits gates we apply are only between $a_i$ and $b_i$ $\forall i$, which are connected.

It holds that
$$
\Lambda(S)_{[a_0,b_0,\ldots,a_{d-1},b_{d-1}]} \otimes A_{[c_0,\ldots,c_{d-1}]}  = \left[\left(\bigotimes_{i=0}^{d-1} \underbrace{\text{SWAP}_{[b_i,c_i]}}_{\text{swap $b_i$ and $c_i$}}\right) \times \left(\bigotimes_{i=0}^{d-1} \underbrace{\text{SWAP}_{[c_i,b_{i+1}]}}_{\text{swap $c_i$ and $b_{i+1}$}}\right)\right] \otimes I_{[a_0,\ldots,a_{d-1}]},
$$
where 
$$
A\ket{0,\ldots,0} = \ket{0,\ldots,0} \quad \text{and} \quad A\ket{\delta_i^n} = \ket{\delta_{i-1}^n}.
$$
We need to apply two swaps, one between $b_i$ and $c_i$, and the other one between $c_i$ and $b_{i+1}$. This respects the connectivity between qubits.

We note 
$$\Gamma(\mathcal{QW}_t) =
\left(\Lambda(S)_{[a_0,b_0,\ldots,a_{d-1},b_{d-1}]} \otimes A_{[c_0,\ldots,c_{d-1}]}\right) \times \left(\Lambda(C_t)_{[a_0,b_0,\ldots,a_{d-1},b_{d-1}]} \otimes I_{[c_0,\ldots,c_{d-1}]}\right).$$

It holds that
\[
\begin{split}
	\cprod{t=t_0}{t_0+d}\Lambda(\mathcal{QW}_t) &= \cprod{t=t_0}{t_0+d}\left(\Lambda(S\times C_t)_{[a_0,b_0,\ldots,a_{d-1},b_{d-1}]} \otimes A_{[c_0,\ldots,c_{d-1}]}\right)\\
	&= \left(\cprod{t=t_0}{t_0+d}\Lambda(S\times C_t)_{[a_0,b_0,\ldots,a_{d-1},b_{d-1}]}\right) \otimes \left(\cprod{t=t_0}{t_0+d} A_{[c_0,\ldots,c_{d-1}]}\right)\\
	&= \Lambda\left(\cprod{t=t_0}{t_0+d}(S\times C_t)\right)_{[a_0,b_0,\ldots,a_{d-1},b_{d-1}]} \otimes I_{[c_0,\ldots,c_{d-1}]}\\
	& \\
	&= \Lambda\left(\mathcal{QW}_{(t_0,t_0+d)}\right)_{[a_0,b_0,\ldots,a_{d-1},b_{d-1}]} \otimes I_{[c_0,\ldots,c_{d-1}]}.\\
\end{split}
\]
Therefor, the unary extension $\Lambda(\mathcal{QW})$ of the quantum walk $\mathcal{QW}$ is defined by
$$
\Lambda(\mathcal{QW}) = \cprod{t=0}{t_f} \Lambda(\mathcal{QW}_t),
$$
and satisfies
$$
\Lambda(\mathcal{QW}) = \Lambda(T_{[a_0,b_0,\ldots,a_{d-1},b_{d-1}]}) \otimes I_{[c_0,\ldots,c_{d-1}]} = 
\Lambda(U_{[a_0,\ldots,a_{d-1}]}) \otimes I_{[b_0,c_0,\ldots,b_{d-1},c_{d-1}]}.
$$
Furthermore, the unary extension $\Lambda(\mathcal{QW})$ admits a circuit of depth $O(d^2)$ that respects the connectivity $\forall i, \; a_i \sim b_i \sim c_i \sim b_{i+1}$ illustrated in Figure \ref{fig:node_register}. This circuit's template is provided in Figure \ref{fig:circuit_scatter}.

\begin{figure}
	\centering
	\Qcircuit @C=1em @R=.7em {
		b_0 	& & \targ 		& \qw 		& \qw 		& \ctrl{7} 		& \qw 			& \qw 			& \qw 		& \qw 		& \targ 	& \qw & \qswap 		& \qw 		 & \qswap 		& \qw \inputgroupv{3}{4}{5em}{0.5em}{\text{Node register}\hspace{9em}}\\
		c_0 	& & \qw 		& \qw 		& \qw 		& \qw 			& \qw 			& \qw 			& \qw 		& \qw 		& \qw 		& \qw & \qswap\qwx 	& \qswap 	 & \qw\qwx 		& \qw \\
		b_i 	& & \qw 		& \targ 	& \qw 		& \qw 			& \ctrl{6} 		& \qw 			& \qw 		& \targ 	& \qw 		& \qw & \qswap 		& \qswap\qwx & \qw\qwx 		& \qw \\
		c_i 	& & \qw 		& \qw 		& \qw 		& \qw 			& \qw 			& \qw 			& \qw 		& \qw 		& \qw 		& \qw & \qswap\qwx 	& \qswap 	 & \qw\qwx 		& \qw \\
		b_{d-1} & & \qw 		& \qw 		& \targ 	& \qw 			& \qw 			& \ctrl{5} 		& \targ 	& \qw 		& \qw 		& \qw & \qswap 		& \qswap\qwx & \qw\qwx 		& \qw \\
		c_{d-1} & & \qw 		& \qw 		& \qw 		& \qw 			& \qw 			& \qw 			& \qw 		& \qw 		& \qw 		& \qw & \qswap\qwx 	& \qw 		 & \qswap\qwx 	& \qw \\
		& & 			& 			& 			& 				& 				& 				& 			& 			& 			& 	  & 			& 			 & 				&		\\
		a_0 	& & \ctrl{-7} 	& \qw 		& \qw 		& \gate{C_t(0)}	& \qw 			& \qw 			& \qw 		& \qw 		& \ctrl{-7} & \qw & \qw 		& \qw 	 	 & \qw 			& \qw \\
		a_i 	& & \qw 		& \ctrl{-6} & \qw 		& \qw 			& \gate{C_t(i)} & \qw 			& \qw 		& \ctrl{-6} & \qw 		& \qw & \qw 		& \qw 	 	 & \qw 			& \qw \inputgroupv{9}{9}{5em}{0em}{\text{Edge register}\hspace{9em}}\\
		a_{d-1} & & \qw 		& \qw 		& \ctrl{-5} & \qw 			& \qw 			& \gate{C_t(d-1)} & \ctrl{-5} & \qw 		& \qw 		& \qw & \qw 		& \qw 	 	 & \qw 			& \qw \\
	}
	\caption{Circuit template for $\Lambda(\mathcal{QW}_t)$}
	\label{fig:circuit_scatter}
\end{figure}
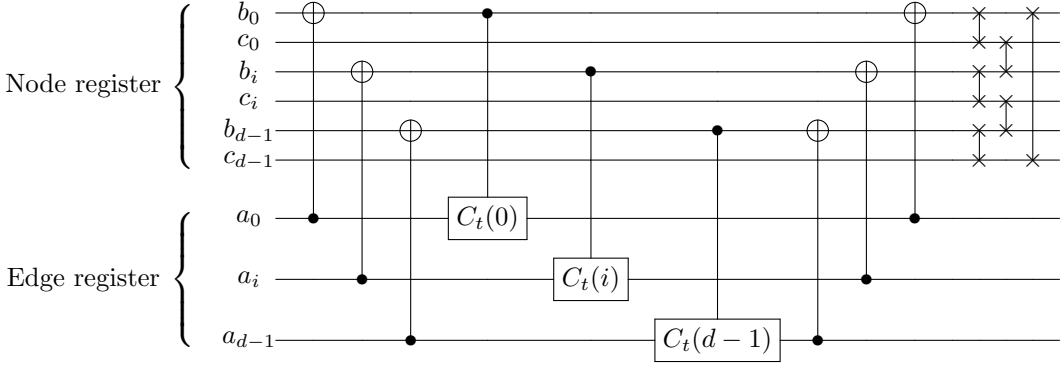

\subsection{Recap}
We introduced a model of distributed quantum computation following the connectivity of a given graph. Two connectivities for the node registers have been considered: all-to-all and cycle. Furthermore, we showed that this model can be used to reproduce the dynamic of the quantum walk on graphs introduced in Section \ref{sec:model}. There are two distributed protocols to reproduce this dynamic, one per connectivity considered for the node registers. Table \ref{tab:distrib} shows a comparison between the two. In this table, the cost (number of two qubits gates) of the circuit inside an edge, the cost of the circuit inside a node and the cost between edges and nodes (for one edge and one node of degree $d$) are all listed. Unsurprisingly, the first method with all-to-all connectivity on the nodes performs better than the QCA-like approach. However, the latter only requires cycle connectivity on the nodes. 

Moreover, the communication cost between two qubits of one edge is always $O(1)$ swaps, and the communication cost between one edge and one node is $O(\ln d) = O(\ln n)$ CNOT for all-to-all connectivity and $O(d) = O(n)$ CNOT for cycle connectivity. Therefore, in general, a cycle connectivity fits the LOCAL model (i.e. $O(\text{poly}(n))$ communication cost between two nodes each round) while one can fit the CONGEST model (i.e. $O(\log n)$ communication cost between two nodes each round) at the cost of all-to-all connectivity on the nodes. Note that this is for general topologies. In practice, a grid or an hypercube fits the CONGEST model even with cycle connectivity.

\begin{table}
	\begin{tabular}{c|ccccc}
		Node connectivity & Edge register & Node register & Edge cost & Node cost & Edge/node cost \\
		\hline
		\rule{0pt}{4ex}  
		all-to-all & $2$ qubits & $O\left(\log d\right)$ qubits & $O\left(1\right)$ & $O\left(d + \mathcal{C}(D)\right)$ & $O\left(\ln d\right)$ CNOT \\
		cycle & $2$ qubits & $O\left(d\right)$ qubits & $O\left(1\right)$ & $O\left(d^2\right)$ & $O\left(d\right)$ SWAP \\
	\end{tabular}
	\caption{Comparison table between the two methods. All information are given for one arbitrary edge or one arbitrary node during one step of the QW's dynamic. $d$ is the degree of said node. $\mathcal{C}(D)$ is the cost (depth) of the diffusion operator $D$. Last column gives the amount of operations between one edge and one node.}
	\label{tab:distrib}
\end{table}
	
	\section{Application: Searching}\label{sec:application}
This section shows how the quantum walk of section \ref{sec:model} can be used for searching. In this section, we only consider the quantum walk and not its distributed implementation presented in section \ref{sec:distrib}. 
This is because the dynamic is the same and studying the mathematical model of quantum walk is somewhat easier.
We first detail how to search edges, before introducing a trick to search nodes. We conclude with some numerical experiments.

\subsection{Searching an edge}
Our model formally describes a quantum walk on the graph's edges. If the quantum state is measured, we obtain one of the edges of the graph. This scheme is thus well suited to search edges. Thus we may now introduce an oracle (i.e. a black box able to recognize/mark the solution to the searching problem) marking the desired edge $(a,b) \in E$. Analogously with the standard spatial search, the oracle is defined as follows~:
$$
\mathcal O_{f} =\left(\sum_{f((u,v))=1}\ket{(u,v)}\bra{(u,v)}\right) \otimes R + \left(\sum_{f((u,v))=0}\ket{(u,v)}\bra{(u,v)}\right)\otimes I_2, 
$$
where $f$ is the classical oracle equals to $1$ if and only if the edge is marked, and $0$ otherwise. The operator $R$ is a special coin operator which is applied only to the marked edge. Without lack of generality, in the following we set $C=X$ and $R=-X$. The algorithm proposed here is the following~: 
\begin{algorithm}
	\caption{Search a marked edge}\label{algo:search}
	\begin{algorithmic}[1]
		\Require $G = (V,E)$ undirected, connected graph
		\Require $f$ the classical oracle
		\Require $T\in \mathbb{N}$ the hitting time
		\Function{Search}{$G,f,T$}
		\State $\displaystyle \ket{\Phi} \gets \frac{1}{\sqrt{2|E|}}\sum_{(u,v)\in E} \ket{(u,v)}\ket{+} + \ket{(u,v)}\ket{-}$ \Comment{Diagonal initial state.}
		\For{$0\leq i < T$}
		\State $\displaystyle \ket{\Phi} \gets S \times (I\otimes C) \times \mathcal O_{f} \times \ket{\Phi}$ \Comment{One step of the walk}
		\EndFor
		\State $(u,v,\pm) \gets$ \Call{Measure}{$\ket{\Phi}$} \Comment{Measure the quantum state.}
		\State \Return $(u,v)$
		\EndFunction
	\end{algorithmic}
\end{algorithm}
The initial state is initialized to be diagonal on the basis states and the additional oracle operation is added to the former QW-dynamic. There are two parameters that characterize the above algorithm~: the probability of success $P$ of returning the marked edge and the hitting time $T$. The number of oracle calls of this algorithm is $O(T)$. In practice, we want to choose $T$ such that $P$ is maximal.
The previous algorithm can easily be made into a Las Vegas algorithm whose answer is always correct but the running time random. Algorithm \ref{algo:guaranteed_search} shows such transformation. Interestingly, the expected number of times Algorithm \ref{algo:guaranteed_search} calls Algorithm \ref{algo:search} is $O(1/P)$ where $P$ is the probability of success of Algorithm \ref{algo:search}. The expected complexity of Algorithm \ref{algo:guaranteed_search} is $O\left(T/P\right).$
The complexity of our algorithm is limited by the optimal complexity of $O(\sqrt{K})$ for searching problem where $K$ is the total number of elements. This is the complexity of the Grover algorithm which has been shown to be optimal \cite{zalka1999grover}. However Grover's algorithm assume full connectivity between elements which is not our case. The quantum walk presented here is on the edge of $G$, which means that the optimal complexity is $O\left(\sqrt{|E|}\right)$.

\begin{algorithm}
	\caption{Search a marked edge with guaranteed success}\label{algo:guaranteed_search}
	\begin{algorithmic}[1]
		\Require $G = (V,E)$ undirected, connected graph
		\Require $f$, the classical oracle
		\Require $T\in \mathbb{N}$ the hitting time
		\Function{GuaranteedSearch}{$G,f,T$}
		\State $(u,v) \gets (\texttt{nil},\texttt{nil})$ \Comment{Initial value}
		\While{$f((u,v)) \neq 1$} \Comment{Until we find the marked edge ...}
		\State $(u,v) \gets$ \Call{Search}{$G,f,T$} \Comment{... search again.}
		\EndWhile
		\State \Return $(u,v)$
		\EndFunction
	\end{algorithmic}
\end{algorithm}

\subsection{Example : Searching an Edge in the Star Graph}
In this example we consider the star graph with $M$ edges and $M+1$ nodes; a graph with $M+1$ nodes where every node (other than node $u_0$) is connected to node $u_0$. Our searching algorithm performs well on this graph, as shown by Theorem \ref{th:star}. In this section, we show the proof of Theorem \ref{th:star} which is mainly spectral analysis on a simplified dynamic.

\begin{theorem}\label{th:star}
	Let $G$ be the star graph with $M$ edges. Algorithm \ref{algo:search} has an optimal hitting time $T=O\left(\sqrt{M}\right)$ and a probability of success $O(1)$ for $G$. Algorithm \ref{algo:guaranteed_search} has expected complexity $O\left(\sqrt{M}\right)$.
\end{theorem}
\begin{proof}
	We consider the star graph $G = (V,E)$ of size $M+1$ such that
	$
	V = \{u_0,\ldots,u_M\} \text{ and } E = \{(u_0,u_i)\mid 1\leq i \leq M\}.
	$
	We assume without a loss of generality that the marked edge is $(u_0,u_1)$ and that the polarity is $\forall i>1, \; \sigma(u_0,u_i) = +$. At any time step $t$, the state of the walk reads $\ket{\Phi_t} = \sum_{i=1}^M \psi_{u_0,u_i}^+(t)\ket{(u_0,u_i)}\ket{+} + \psi_{u_0,u_i}^-(t)\ket{(u_0,u_i)}\ket{-}.$	
	We first show that $\forall t\in \mathbb{N}, \; \forall i > 1, \; \psi_{u_0,u_i}^+(t) = \alpha^+ \text{ and } \psi_{u_0,u_i}^-(t) = \alpha^-.$ This greatly simplifies the way we describe the walk dynamic. Afterward we shall provide simple spectral analysis to extract the optimal hitting time $T$ and probability of success $P$. Next, we prove 
	the property $(Q_t): \quad \forall i > 1, \; \psi_{u_0,u_i}^+(t) = \alpha_t^+ \text{ and } \psi_{u_0,u_i}^-(t) = \alpha_t^-$ for all $t\in \mathbb{N}$.
	
	The initial state $\ket{\Psi_0}$ is diagonal on the basis states~:
	$$ \ket{\Phi} \gets \frac{1}{\sqrt{2|E|}}\sum_{(u,v)\in E} \ket{(u,v)}\ket{+} + \ket{(u,v)}\ket{-}.$$
	All $\left(\psi_{u_0,u_i}^\pm(0)\right)_{i\geq 1}$ are equals so the property $Q_0$ is satisfied with $\alpha_t^+ = \alpha_t^- = \frac{1}{\sqrt{2M}}$. Now, we assume that $Q_t$ is true and use the walk dynamic to show that $Q_{t+1}$ is also true. The state $\ket{\Psi_{t+1}}$ is described in Table \ref{tab:star} and show that $Q_{t+1}$ is true.
	
	\begin{table}[h]
		\caption{Detailed dynamic of the quantum walk for star graphs.}
		\label{tab:star}
		\begin{tabular}{l|cccl}
			& Step $t$ & After oracle & After coin & After Scattering \\
			\hline
			& & & & \\
			$\alpha^+$ & $\alpha^+_t$ & $\alpha^+_t$ & $\alpha^-_t$ & $\displaystyle\alpha^+_{t+1} = \frac{(M-2)\alpha_t^- - 2\psi_{u_0,u_1}^+}{M}(t)$ \\
			& & & & \\
			$\alpha^-$ & $\alpha^-_t$ & $\alpha^-_t$ & $\alpha^+_t$ & $\alpha^-_{t+1}=\alpha^+_t$ \\
			& & & & \\
			$\psi_{u_0,u_1}^+$ & $\psi_{u_0,u_1}^+(t)$ & $-\psi_{u_0,u_1}^-(t)$ & $-\psi_{u_0,u_1}^+(t)$ & $\displaystyle\psi_{u_0,u_1}^+(t+1) = \frac{(2M-2)\alpha_t^- +(M-2)\psi_{u_0,u_1}^+}{M}(t)$ \\
			& & & & \\
			$\psi_{u_0,u_1}^-$ & $\psi_{u_0,u_1}^-(t)$ & $-\psi_{u_0,u_1}^+(t)$ & $-\psi_{u_0,u_1}^-(t)$ & $\psi_{u_0,u_1}^-(t+1)=-\psi_{u_0,u_1}^-(t)$ \\
		\end{tabular}
	\end{table}
	
	Using the recurrence in Table \ref{tab:star}, we can put the dynamic of the walk into a matrix form.
	$$
	\psi_{u_0,u_1}^-(t) = \frac{(-1)}{\sqrt{2M}} \qquad \text{and} \qquad X_{t+1} = AX_t,
	$$
	where
	$$
	X_t = \mat{\alpha^+_t\\\alpha_t^-\\\psi_{u_0,u_1}^+(t)\\}\qquad \text{and} \qquad A = \mat{
		0 & \frac{M-2}{M} & \frac{-2}{M}\\
		1 & 0 & 0\\
		0 & 2\frac{M-1}{M} & \frac{M-2}{M}\\
	}.
	$$
	
	After diagonalizing $A$, the eigenvalues are $\{-1,e^{i\lambda},e^{-i\lambda}\}$, where $e^{i\lambda} = \frac{M-1+i\sqrt{2M-1}}{M}$. This leads to
	$
	\psi_{u_0,u_1}^+(t) \sim \frac{i}{2}\left(e^{-i\lambda t} - e^{i\lambda t}\right) \sim \sin(\lambda t)$, which allows us to deduce the probability $p_t$ of hitting the marked edge
	$
	p_t \sim \sin^2(\lambda t),
	$
	and then the optimal hitting time by solving $p_T = 1$:
	$$
	T = \frac{\pi}{2} \frac{1}{\lambda} \sim \frac{\pi}{2} \sqrt{\frac{M}{2}} \sim \frac{\pi}{2\sqrt{2}} \sqrt{M} = O(\sqrt{M}).
	$$
	Finally we have a probability of success $P\sim 1$ for the optimal hitting time $\displaystyle T \sim \frac{\pi}{2\sqrt{2}} \sqrt{M} = O(\sqrt{M})$.
	
\end{proof}

\subsection{Searching Nodes}
The searching quantum walk presented in the previous section can search one marked edge in a graph. In order to search a node instead, we need to transform the graph we walk on. We call this transformation \textsc{starify}.

\begin{definition}{\textsc{Starify}}\\
	Let us consider an undirected graph $G = (V,E)$. The transformation \textsc{starify} on $G$ returns a graph $\tilde G$ for which
	\begin{itemize}
		\item every node and edge in $G$ is in $\tilde G$ (we call them real nodes and real edges).
		\item for every node $u\in V$, we add a node $\tilde u$ (we call these new nodes virtual nodes).
		\item for every node $u\in V$, we add the edge $(u,\tilde u)$ (we call these new edges virtual edges).
	\end{itemize}
	We call the resulting graph $\tilde G$ the starified graph of $G$.
\end{definition}

\paragraph{Searching nodes} Starifying a graph $G$ allows us to search a marked node $u$ instead of an edge. We can then use the previous searching walk to search the virtual edge and then deduce without ambiguity the marked node $u$ of the initial graph $G$. This procedure implies that we have to increase the size of the graph (number of edges and nodes). In particular, increasing the number of edges is significant since we increase the dimension of the walker (this can be significant for sparse graphs) and one must be careful of this when computing the complexity. A reassuring result is that searching a node on the complete graph (the strongest possible connectivity) has optimal complexity. As stated in Theorem \ref{th:complete}, the complexity is $O(\sqrt{M})$, which is optimal when searching over the edges. Compared to a classical algorithm in $O(M)$ (a depth first search for instance), we do have a quadratic speedup.

\begin{theorem}\label{th:complete}
	When using the quantum algorithm to search one marked node in the starified graph $\tilde G$ of the complete graph $G$ of size $N$, the probability of success is $P \sim 1$ and the hitting time $T \sim \frac{\pi}{4} N$.
\end{theorem} 
\begin{proof}
	The proof is very similar to the one of Theorem \ref{th:star}. We show that several edges have the same states to reduce the dynamic to a simple recursive equation. We then solve it numerically and derive asymptotic values for the probability of success.
\end{proof}

\subsection{Numerical experiments}
This section presents numerical results when searching nodes with our model on various graphs: grids, hypercubes, complete graphs and random scale-free graphs. The codes to reproduce the figures in this section is available on a github\footnote{https://github.com/mroget/code\_paper\_QWDistrib}.

\paragraph{Searching Searching nodes in a grid}
In this section we consider the 2-dimensional square grid. Figure \ref{fig:grid} shows the probability of success and hitting time of the searching algorithm in function of $\tilde M$ for several grid sizes. For a grid of size $\sqrt{N}\times \sqrt{N}$, it holds that $\tilde M = 3N$. Linear fit are displayed in Figure \ref{fig:grid} and show numerically that the probability of success $P = O(\frac{1}{\ln \tilde M})$ and the hitting time $T = O(\sqrt{\tilde M \ln \tilde M})$. Since $\tilde M$ and $N$ differ only by a coefficient, the final complexity of the algorithm shall be $O\left(\sqrt{N\log^\frac{3}{2} N}\right)$.

\begin{figure}
	\centering
	\stepcounter{figure}
	\begin{minipage}[t]{0.45\columnwidth}
		\begin{minipage}[c]{\columnwidth}
			\includegraphics[width=\linewidth]{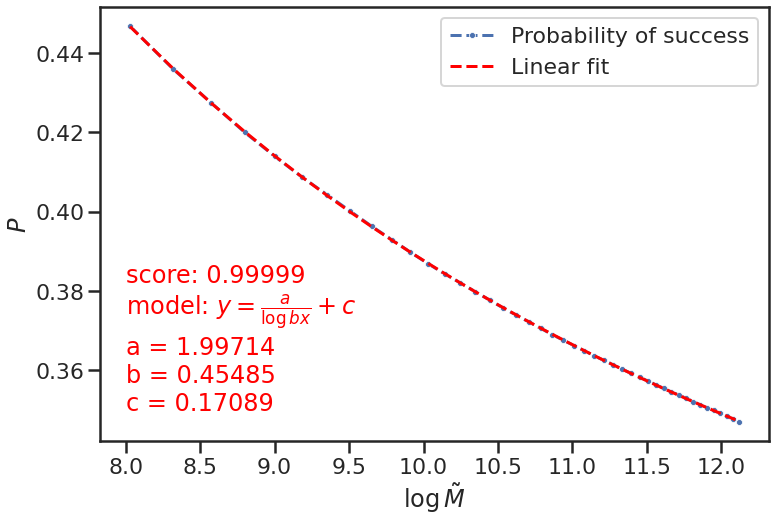}\par
		\end{minipage}
		\Scaption{Probability of success}
	\end{minipage}
	\hfill
	\begin{minipage}[t]{0.45\columnwidth}
		\begin{minipage}[c]{\columnwidth}
			\includegraphics[width=\linewidth]{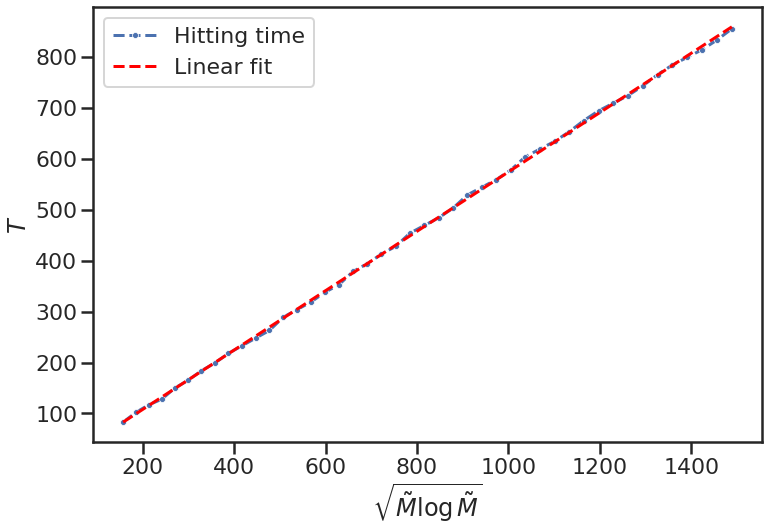}
		\end{minipage}
		\Scaption{Hitting time}
	\end{minipage}
	\addtocounter{figure}{-1}
	\caption{Searching performances for a grid of size $\sqrt{N}\times\sqrt{N}$ when $N$ increases. Linear fits are displayed to exhibit the asymptotic behavior of the searching algorithm. We recall that $\tilde M$ is the number of edges of starified graph.}
	\label{fig:grid}
\end{figure}

\paragraph{Application: Searching nodes in a hypercube}
In this section we consider the hypercube of dimension $d$. Figure \ref{fig:hypercube} shows the probability of success and hitting time of the searching algorithm in function of $\tilde M$ for several values of $d$. It holds that $\tilde M = 2^{d-1}(2+d) = O(N\ln N)$.
Linear fit are displayed in Figure \ref{fig:hypercube} and show numerically that the probability of success $P = \Omega(1)$ and the hitting time $T = O(\sqrt{\tilde M})$. Since $\tilde M = O(N\ln N)$ , the final complexity of the algorithm shall be $O\left(\sqrt{N\log^\star N}\right)$. Note the the quantum walk is walking on the $\tilde M$  edges of the starified graph. Having a complexity of $O(\tilde M)$ means being optimal (in the sense that we do as well as Grover and can't do any better by searching on edges). However, focusing on the original problem of finding a node in the original graph, the complexity becomes slightly worse than Grover.

\begin{figure}
	\centering
	\stepcounter{figure}
	\begin{minipage}[t]{0.45\columnwidth}
		\begin{minipage}[c]{\columnwidth}
			\includegraphics[width=\linewidth]{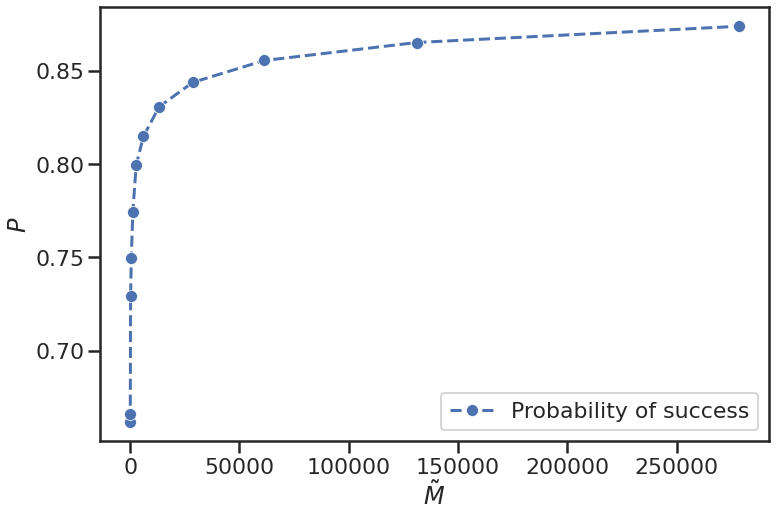}\par
		\end{minipage}
		\Scaption{Probability of success}
	\end{minipage}
	\hfill
	\begin{minipage}[t]{0.45\columnwidth}
		\begin{minipage}[c]{\columnwidth}
			\includegraphics[width=\linewidth]{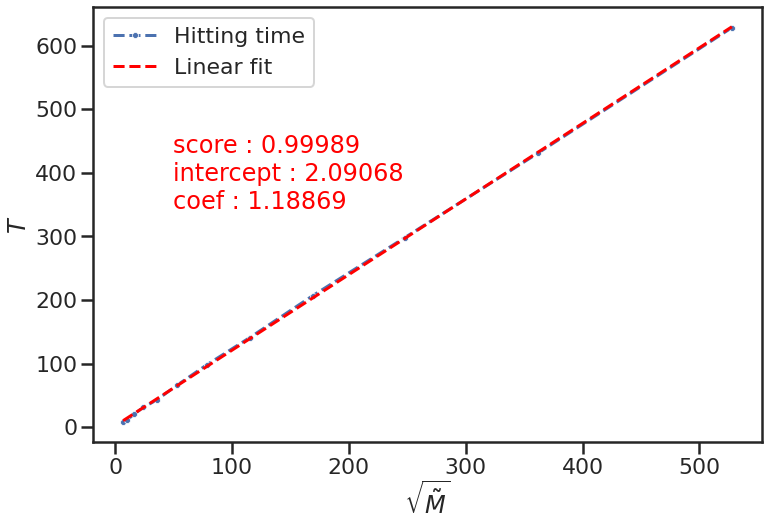}
		\end{minipage}
		\Scaption{Hitting time}
	\end{minipage}
	\addtocounter{figure}{-1}
	\caption{Searching performances for a hypercube of size $N=2^d$ when $N$ increases. Linear fits are displayed to exhibit the asymptotic behavior of the searching algorithm. We recall that $\tilde M$ is the number of edges of starified graph.}
	\label{fig:hypercube}
\end{figure} 

\paragraph{Application: Searching nodes in a complete graph}
In this section we consider the complete graph of size $N$. Figure \ref{fig:complete} shows the probability of success and hitting time of the searching algorithm in function of $\tilde M$ for several values of $N$. It holds that $\tilde M = N + N^2/4$.
Linear fit are displayed in Figure \ref{fig:complete} and show numerically that the probability of success $P = \Omega(1)$ and the hitting time $T = O(\sqrt{\tilde M})$. Since $\tilde M = O(N^2)$ , the final complexity of the algorithm shall be $O\left(N\right)$. Similarly to the hypercube, the QWSearch is optimal.

\begin{figure}
	\centering
	\stepcounter{figure}
	\begin{minipage}[t]{0.45\columnwidth}
		\begin{minipage}[c]{\columnwidth}
			\includegraphics[width=\linewidth]{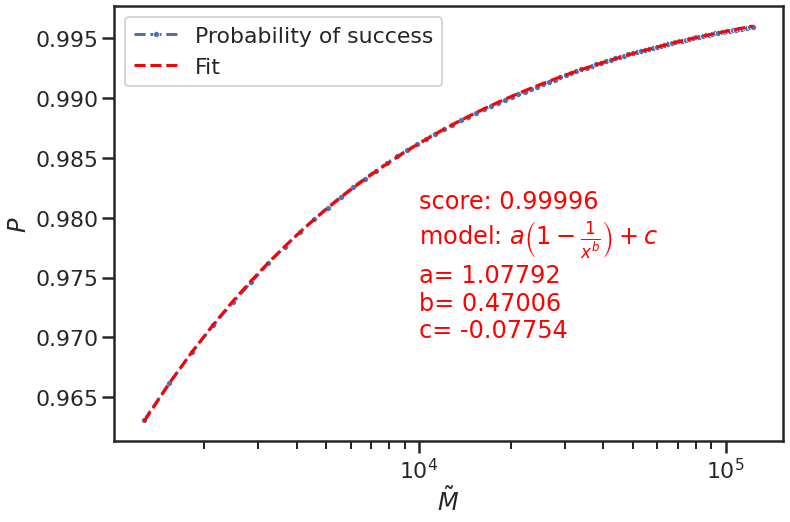}\par
		\end{minipage}
		\Scaption{Probability of success}
	\end{minipage}
	\hfill
	\begin{minipage}[t]{0.45\columnwidth}
		\begin{minipage}[c]{\columnwidth}
			\includegraphics[width=\linewidth]{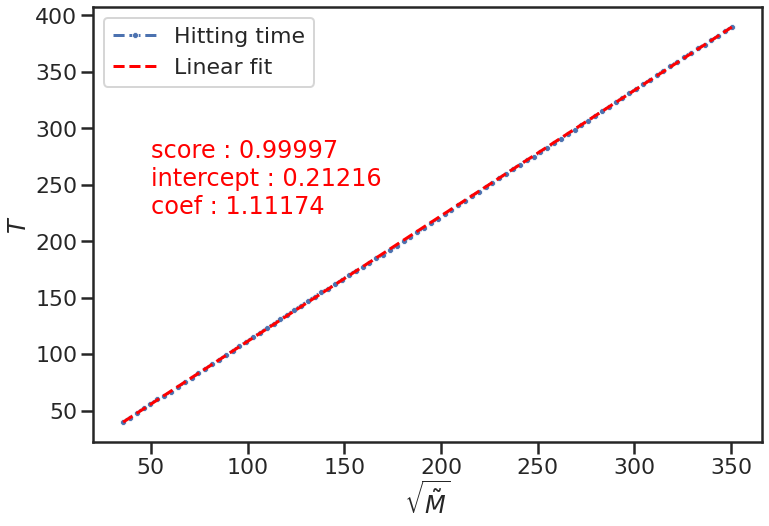}
		\end{minipage}
		\Scaption{Hitting time}
	\end{minipage}
	\addtocounter{figure}{-1}
	\caption{Searching performances for a complete graph of size $N$ when $N$ increases. Linear fits are displayed to exhibit the asymptotic behavior of the searching algorithm. We recall that $\tilde M$ is the number of edges of starified graph.}
	\label{fig:complete}
\end{figure}

\paragraph{Searching nodes in random scale-free graphs}
In this section we focus on random scale-free graphs. Figure \ref{fig:scalefree} shows the probability of success and hitting time for every nodes of one random scale-free graph generated with the Barabási–Albert algorithm\cite{barabasi1999emergence}. The fit suggests the following relations:
$$
P = 1 - \frac{1}{1+ \frac{3}{4} N \tan \frac{2d}{\pi(N-1)}} 
\qquad  \text{and} \qquad 
T^2 \propto \frac{N+M}{P},
$$
where $d$ is the degree of the node marked. To confirm these relations, we randomly generated 1000 scale-free bipartite graphs (the parameters $n,m$ have been randomly sampled between $100$ and $200$ for $n$ and between $5n$ and $n^2/2$ for $m$). We then look at the mean square error and the coefficient of determination obtained when predicting the probability of success and hitting time for every nodes of every graphs. These results are presented in Table \ref{tab:fit}.
\begin{table}[h]
	\centering
	\begin{tabular}{c|cc}
		& \textbf{mse}  & \textbf{score} \\
		\hline
		$T$ & $0.5237187281594099$ & $0.9985740735579267$\\
		$P$ & $2.5618806822927392e-06$ & $0.9986901398080446$\\
	\end{tabular}
	\caption{Mean Square Error (mse) and coefficient of determination (score) when trying to fit the probability of success and hitting time.}
	\label{tab:fit}
\end{table}

\begin{figure}
	\centering
	\stepcounter{figure}
	\begin{minipage}[t]{0.4\columnwidth}
		\begin{minipage}[c]{\columnwidth}
			\includegraphics[width=\linewidth]{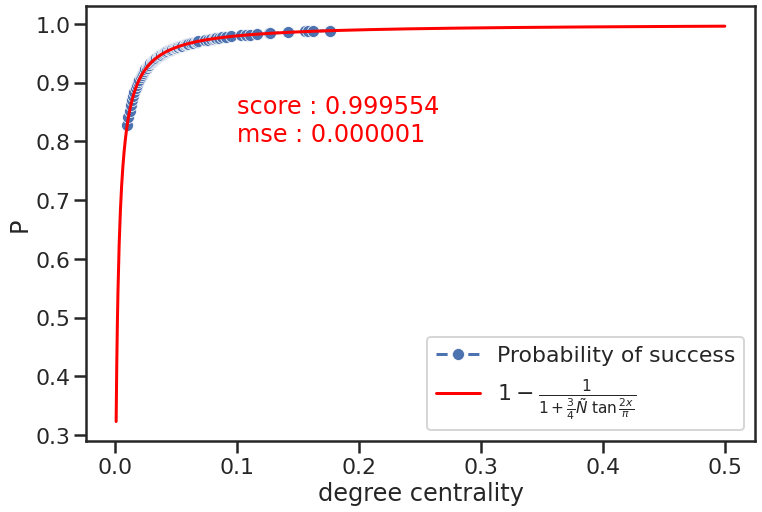}\par
		\end{minipage}
		\Scaption{Probability of success}
	\end{minipage}
	\hfill
	\begin{minipage}[t]{0.4\columnwidth}
		\begin{minipage}[c]{\columnwidth}
			\includegraphics[width=\linewidth]{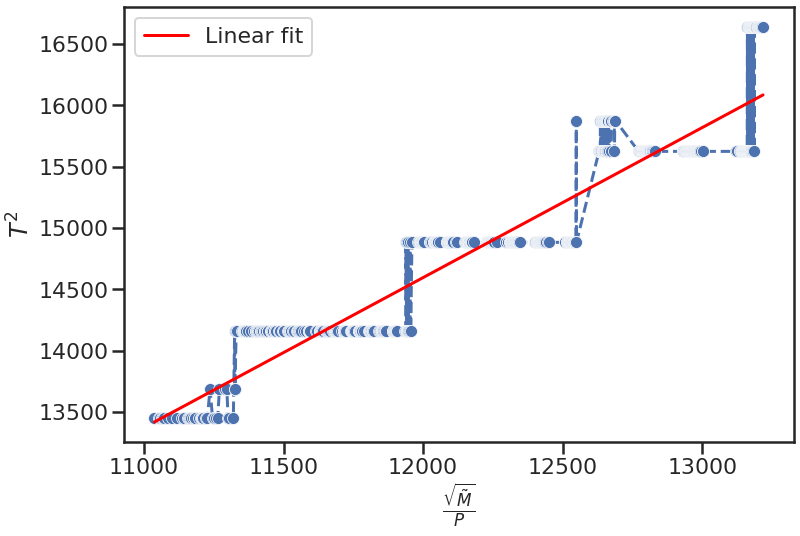}
		\end{minipage}
		\Scaption{Hitting time}
	\end{minipage}
	\addtocounter{figure}{-1}
	\caption{Searching performances for a random scale-free graph of size $1000$ with around $9900$ edges. Fits are displayed to exhibit the asymptotic behavior of the searching algorithm. We recall that $\tilde M$ is the number of edges of starified graph.}
	\label{fig:scalefree}
\end{figure}

	\section*{Conclusion}\label{ccl}
	
	In this paper, we presented a new distributed model of quantum computation inspired by quantum cellular automata. This model preserves the nice properties that QCA have, have, such as locality or translation invariance. Furthermore, this model is defined on arbitrary topologies. A network of qubits is established, with a register of two qubits per edge and a multiple qubits register per node. Communications, in this model, are counted as the number of two-qubits gates applied. Moreover, two connectivities are considered for the qubits of one node: all-to-all and cycle. As an application, we show how this model can simulate the dynamic of a quantum walk, as in~\cite{roget2024quantum}. This is a natural choice as a quantum walk formally coincides with the single particle sector of a QCA. Two different protocols reproducing the quantum walk dynamic are considered, one for each node registers' connectivities. Obviously, all-to-all connectivity has a smaller communication cost than cycle and always fits the CONGEST model of communication. However, even cycle connectivity fits the CONGEST model for several topologies like grids or hypercubes.
	
	Finally, we present an application of this quantum walk to the searching problem. We show how the dynamic of the quantum walk can be used to search a node and an edge on several graphs. Numerical results show a quadratic improvement of the complexity for almost all graphs tested. While these numerical results are for the quantum walk's dynamic, it can be translated to a distributed framework since the distributed scheme from Section \ref{sec:distrib} can perfectly reproduce said dynamic.
	
	To go further, one could consider other dynamics besides the quantum walk. Indeed, while our model of computation has been inspired by quantum walks and implements them well, it is more general than that. Further works could also been done on exploiting the searching quantum walk in a distributed framework. 
	
	\paragraph{Aknowledgement}
	This work is supported by the PEPR EPiQ ANR-22-PETQ-0007, by the ANR JCJC DisQC ANR-22-CE47-0002-01.
	
	\bibliography{sn-bibliography}


\begin{thebibliography}{18}
\ifx \bisbn   \undefined \def \bisbn  #1{ISBN #1}\fi
\ifx \binits  \undefined \def \binits#1{#1}\fi
\ifx \bauthor  \undefined \def \bauthor#1{#1}\fi
\ifx \batitle  \undefined \def \batitle#1{#1}\fi
\ifx \bjtitle  \undefined \def \bjtitle#1{#1}\fi
\ifx \bvolume  \undefined \def \bvolume#1{\textbf{#1}}\fi
\ifx \byear  \undefined \def \byear#1{#1}\fi
\ifx \bissue  \undefined \def \bissue#1{#1}\fi
\ifx \bfpage  \undefined \def \bfpage#1{#1}\fi
\ifx \blpage  \undefined \def \blpage #1{#1}\fi
\ifx \burl  \undefined \def \burl#1{\textsf{#1}}\fi
\ifx \doiurl  \undefined \def \doiurl#1{\url{https://doi.org/#1}}\fi
\ifx \betal  \undefined \def \betal{\textit{et al.}}\fi
\ifx \binstitute  \undefined \def \binstitute#1{#1}\fi
\ifx \binstitutionaled  \undefined \def \binstitutionaled#1{#1}\fi
\ifx \bctitle  \undefined \def \bctitle#1{#1}\fi
\ifx \beditor  \undefined \def \beditor#1{#1}\fi
\ifx \bpublisher  \undefined \def \bpublisher#1{#1}\fi
\ifx \bbtitle  \undefined \def \bbtitle#1{#1}\fi
\ifx \bedition  \undefined \def \bedition#1{#1}\fi
\ifx \bseriesno  \undefined \def \bseriesno#1{#1}\fi
\ifx \blocation  \undefined \def \blocation#1{#1}\fi
\ifx \bsertitle  \undefined \def \bsertitle#1{#1}\fi
\ifx \bsnm \undefined \def \bsnm#1{#1}\fi
\ifx \bsuffix \undefined \def \bsuffix#1{#1}\fi
\ifx \bparticle \undefined \def \bparticle#1{#1}\fi
\ifx \barticle \undefined \def \barticle#1{#1}\fi
\bibcommenthead
\ifx \bconfdate \undefined \def \bconfdate #1{#1}\fi
\ifx \botherref \undefined \def \botherref #1{#1}\fi
\ifx \url \undefined \def \url#1{\textsf{#1}}\fi
\ifx \bchapter \undefined \def \bchapter#1{#1}\fi
\ifx \bbook \undefined \def \bbook#1{#1}\fi
\ifx \bcomment \undefined \def \bcomment#1{#1}\fi
\ifx \oauthor \undefined \def \oauthor#1{#1}\fi
\ifx \citeauthoryear \undefined \def \citeauthoryear#1{#1}\fi
\ifx \endbibitem  \undefined \def \endbibitem {}\fi
\ifx \bconflocation  \undefined \def \bconflocation#1{#1}\fi
\ifx \arxivurl  \undefined \def \arxivurl#1{\textsf{#1}}\fi
\csname PreBibitemsHook\endcsname

\bibitem[\protect\citeauthoryear{Arrighi}{2019}]{arrighi2019overview}
\begin{barticle}
\bauthor{\bsnm{Arrighi}, \binits{P.}}:
\batitle{An overview of quantum cellular automata}.
\bjtitle{Natural Computing}
\bvolume{18}(\bissue{4}),
\bfpage{885}--\blpage{899}
(\byear{2019})
\end{barticle}
\endbibitem

\bibitem[\protect\citeauthoryear{Di~Molfetta et~al.}{2013}]{di2013quantum}
\begin{barticle}
\bauthor{\bsnm{Di~Molfetta}, \binits{G.}},
\bauthor{\bsnm{Brachet}, \binits{M.}},
\bauthor{\bsnm{Debbasch}, \binits{F.}}:
\batitle{Quantum walks as massless dirac fermions in curved space-time}.
\bjtitle{Physical Review A}
\bvolume{88}(\bissue{4}),
\bfpage{042301}
(\byear{2013})
\end{barticle}
\endbibitem

\bibitem[\protect\citeauthoryear{Bisio et~al.}{2016}]{bisio2016quantum}
\begin{barticle}
\bauthor{\bsnm{Bisio}, \binits{A.}},
\bauthor{\bsnm{D'Ariano}, \binits{G.M.}},
\bauthor{\bsnm{Perinotti}, \binits{P.}}:
\batitle{Quantum walks, deformed relativity and hopf algebra symmetries}.
\bjtitle{Philosophical Transactions of the Royal Society A: Mathematical,
  Physical and Engineering Sciences}
\bvolume{374}(\bissue{2068}),
\bfpage{20150232}
(\byear{2016})
\end{barticle}
\endbibitem

\bibitem[\protect\citeauthoryear{Santha}{2008}]{santha2008quantum}
\begin{bchapter}
\bauthor{\bsnm{Santha}, \binits{M.}}:
\bctitle{Quantum walk based search algorithms}.
In: \bbtitle{International Conference on Theory and Applications of Models of
  Computation},
pp. \bfpage{31}--\blpage{46}
(\byear{2008}).
\bcomment{Springer}
\end{bchapter}
\endbibitem

\bibitem[\protect\citeauthoryear{Slate et~al.}{2021}]{slate2021quantum}
\begin{barticle}
\bauthor{\bsnm{Slate}, \binits{N.}},
\bauthor{\bsnm{Matwiejew}, \binits{E.}},
\bauthor{\bsnm{Marsh}, \binits{S.}},
\bauthor{\bsnm{Wang}, \binits{J.B.}}:
\batitle{Quantum walk-based portfolio optimisation}.
\bjtitle{Quantum}
\bvolume{5},
\bfpage{513}
(\byear{2021})
\end{barticle}
\endbibitem

\bibitem[\protect\citeauthoryear{Melnikov
  et~al.}{2019}]{melnikov2019predicting}
\begin{barticle}
\bauthor{\bsnm{Melnikov}, \binits{A.A.}},
\bauthor{\bsnm{Fedichkin}, \binits{L.E.}},
\bauthor{\bsnm{Alodjants}, \binits{A.}}:
\batitle{Predicting quantum advantage by quantum walk with convolutional neural
  networks}.
\bjtitle{New Journal of Physics}
\bvolume{21}(\bissue{12}),
\bfpage{125002}
(\byear{2019})
\end{barticle}
\endbibitem

\bibitem[\protect\citeauthoryear{Gall et~al.}{2018}]{gall2018quantum}
\begin{botherref}
\oauthor{\bsnm{Gall}, \binits{F.L.}},
\oauthor{\bsnm{Nishimura}, \binits{H.}},
\oauthor{\bsnm{Rosmanis}, \binits{A.}}:
Quantum advantage for the local model in distributed computing.
arXiv preprint arXiv:1810.10838
(2018)
\end{botherref}
\endbibitem

\bibitem[\protect\citeauthoryear{Le~Gall and Magniez}{2018}]{le2018sublinear}
\begin{bchapter}
\bauthor{\bsnm{Le~Gall}, \binits{F.}},
\bauthor{\bsnm{Magniez}, \binits{F.}}:
\bctitle{Sublinear-time quantum computation of the diameter in congest
  networks}.
In: \bbtitle{Proceedings of the 2018 ACM Symposium on Principles of Distributed
  Computing},
pp. \bfpage{337}--\blpage{346}
(\byear{2018})
\end{bchapter}
\endbibitem

\bibitem[\protect\citeauthoryear{Izumi et~al.}{2019}]{izumi2019quantum}
\begin{botherref}
\oauthor{\bsnm{Izumi}, \binits{T.}},
\oauthor{\bsnm{Gall}, \binits{F.L.}},
\oauthor{\bsnm{Magniez}, \binits{F.}}:
Quantum distributed algorithm for triangle finding in the congest model.
arXiv preprint arXiv:1908.11488
(2019)
\end{botherref}
\endbibitem

\bibitem[\protect\citeauthoryear{Bezerra et~al.}{2021}]{bezerra2021quantum}
\begin{barticle}
\bauthor{\bsnm{Bezerra}, \binits{G.}},
\bauthor{\bsnm{Lug{\~a}o}, \binits{P.H.}},
\bauthor{\bsnm{Portugal}, \binits{R.}}:
\batitle{Quantum-walk-based search algorithms with multiple marked vertices}.
\bjtitle{Physical Review A}
\bvolume{103}(\bissue{6}),
\bfpage{062202}
(\byear{2021})
\end{barticle}
\endbibitem

\bibitem[\protect\citeauthoryear{Arrighi et~al.}{2018}]{arrighi2018dirac}
\begin{barticle}
\bauthor{\bsnm{Arrighi}, \binits{P.}},
\bauthor{\bsnm{Di~Molfetta}, \binits{G.}},
\bauthor{\bsnm{M{\'a}rquez-Mart{\'\i}n}, \binits{I.}},
\bauthor{\bsnm{P{\'e}rez}, \binits{A.}}:
\batitle{Dirac equation as a quantum walk over the honeycomb and triangular
  lattices}.
\bjtitle{Physical Review A}
\bvolume{97}(\bissue{6}),
\bfpage{062111}
(\byear{2018})
\end{barticle}
\endbibitem

\bibitem[\protect\citeauthoryear{Roget et~al.}{2020}]{roget2020grover}
\begin{barticle}
\bauthor{\bsnm{Roget}, \binits{M.}},
\bauthor{\bsnm{Guillet}, \binits{S.}},
\bauthor{\bsnm{Arrighi}, \binits{P.}},
\bauthor{\bsnm{Di~Molfetta}, \binits{G.}}:
\batitle{Grover search as a naturally occurring phenomenon}.
\bjtitle{Physical Review Letters}
\bvolume{124}(\bissue{18}),
\bfpage{180501}
(\byear{2020})
\end{barticle}
\endbibitem

\bibitem[\protect\citeauthoryear{Portugal}{2013}]{portugal2013quantum}
\begin{bbook}
\bauthor{\bsnm{Portugal}, \binits{R.}}:
\bbtitle{Quantum Walks and Search Algorithms}
vol. \bseriesno{19}.
\bpublisher{Springer}, \blocation{???}
(\byear{2013})
\end{bbook}
\endbibitem

\bibitem[\protect\citeauthoryear{Biswal et~al.}{2019}]{biswal2019techniques}
\begin{barticle}
\bauthor{\bsnm{Biswal}, \binits{L.}},
\bauthor{\bsnm{Bhattacharjee}, \binits{D.}},
\bauthor{\bsnm{Chattopadhyay}, \binits{A.}},
\bauthor{\bsnm{Rahaman}, \binits{H.}}:
\batitle{Techniques for fault-tolerant decomposition of a multicontrolled
  toffoli gate}.
\bjtitle{Physical Review A}
\bvolume{100}(\bissue{6}),
\bfpage{062326}
(\byear{2019})
\end{barticle}
\endbibitem

\bibitem[\protect\citeauthoryear{Nielsen and Chuang}{2010}]{nielsen2010quantum}
\begin{bbook}
\bauthor{\bsnm{Nielsen}, \binits{M.A.}},
\bauthor{\bsnm{Chuang}, \binits{I.L.}}:
\bbtitle{Quantum Computation and Quantum Information}.
\bpublisher{Cambridge university press}, \blocation{???}
(\byear{2010})
\end{bbook}
\endbibitem

\bibitem[\protect\citeauthoryear{Zalka}{1999}]{zalka1999grover}
\begin{barticle}
\bauthor{\bsnm{Zalka}, \binits{C.}}:
\batitle{Grover’s quantum searching algorithm is optimal}.
\bjtitle{Physical Review A}
\bvolume{60}(\bissue{4}),
\bfpage{2746}
(\byear{1999})
\end{barticle}
\endbibitem

\bibitem[\protect\citeauthoryear{Barab{\'a}si and
  Albert}{1999}]{barabasi1999emergence}
\begin{barticle}
\bauthor{\bsnm{Barab{\'a}si}, \binits{A.-L.}},
\bauthor{\bsnm{Albert}, \binits{R.}}:
\batitle{Emergence of scaling in random networks}.
\bjtitle{science}
\bvolume{286}(\bissue{5439}),
\bfpage{509}--\blpage{512}
(\byear{1999})
\end{barticle}
\endbibitem

\bibitem[\protect\citeauthoryear{Roget and
  Di~Molfetta}{2024}]{roget2024quantum}
\begin{bchapter}
\bauthor{\bsnm{Roget}, \binits{M.}},
\bauthor{\bsnm{Di~Molfetta}, \binits{G.}}:
\bctitle{A quantum walk-based scheme for distributed searching on arbitrary
  graphs}.
In: \bbtitle{Asian Symposium on Cellular Automata Technology},
pp. \bfpage{72}--\blpage{83}
(\byear{2024}).
\bcomment{Springer}
\end{bchapter}
\endbibitem

\end{thebibliography}
	
\end{document}